\documentclass[12pt]{article}
\date{}
\usepackage{amsfonts}
\usepackage{amsmath}
\usepackage{graphicx,framed}
\usepackage{xcolor}

\usepackage[a4paper]{geometry}
\usepackage{amsmath}
\usepackage{amssymb}
\usepackage{amsopn,ntheorem}
\usepackage[font=footnotesize]{caption}
\usepackage{subcaption}

    \newcommand{\bigO}[1]{\ensuremath{\mathop{}\mathopen{}\mathcal{O}\mathopen{}\left(#1\right)}}

\newenvironment{ybox}{
  \setlength{\FrameSep}{1mm}
  \setlength{\FrameRule}{0mm}
  
  \MakeFramed {\FrameRestore}}
{\endMakeFramed}

\newenvironment{bbox}{
	  \setlength{\FrameSep}{1mm}
	  \setlength{\FrameRule}{0mm}
  
  \MakeFramed {\FrameRestore}}
{\endMakeFramed}

\newcommand{\ps}{\mathsf{p}} 
\newcommand{\tr}{\mathtt{Tr}}
\newcommand{\X}{\textbf{X}}
\newcommand{\Y}{\textbf{Y}}
\newcommand{\Z}{\textbf{Z}}

\newcommand{\diag}{\mathsf{diag}}

\newtheorem{theorem}{Theorem}%[section]
\newtheorem{lemma}{Lemma}

\newtheorem{proposition}{Proposition}

\newtheorem{example}{Example}
\newtheorem{discussion}{Discussion}

\newenvironment{proof}{{\noindent{\bf Proof:}}}{$\hfill\Box$}
\allowdisplaybreaks
\begin{document}
\title{Sampling and Distortion Tradeoffs for Bandlimited Periodic Signals}

\author{Elaheh Mohammadi and Farokh Marvasti\\\small Advanced Communications Research Institute (ACRI) \\ \small Department of Electrical Engineering\\\small
Sharif University of Technology\footnote{This paper was presented in part at the 2015 Sampling Theory and Applications
(SampTA) conference.}}

\maketitle
\vspace{-0.8cm}
\begin{abstract} In this paper, the optimal sampling strategies  (uniform or nonuniform) and distortion tradeoffs for Gaussian bandlimited periodic signals with additive white Gaussian noise are studied.  Our emphasis is on characterizing the optimal sampling locations as well as the optimal pre-sampling filter to minimize the reconstruction distortion.  We first show that to achieve the optimal distortion, no pre-sampling filter is necessary for any arbitrary sampling rate. Then, we provide a complete characterization of optimal distortion for low and high sampling rates (with respect to the signal bandwidth). We also provide bounds on the reconstruction distortion for rates in the intermediate region. It is shown that nonuniform sampling outperforms uniform sampling for low sampling rates. In addition, the optimal nonuniform sampling set is robust with respect to missing sampling values. On the other hand, for the sampling rates above the Nyquist rate, the uniform sampling strategy is optimal.
An extension of the results for random discrete periodic signals is discussed with simulation results indicating that the intuitions from the continuous domain carry over to the discrete domain. Sparse signals are also considered, where it is shown that uniform sampling is optimal above the Nyquist rate. 
\end{abstract}

\section{Introduction}
Shannon's rate distortion theory finds the optimal compression rate for a given distortion for an i.i.d.\ discrete signal. For continuous signals, Shannon assumes that the signal is sampled  (noiselessly)  at the Nyquist rate. Thus, in Shannon's protocol, distortion occurs at the quantization phase; no information about the signal is discarded in the sampling phase. But since only the end-to-end distortion is of importance, one can accept some distortion at the sampling phase, by sampling the signal below the Nyquist rate, or by assuming sampling noise. In other words, we can discard information about the signal in both the sampling and quantization phases. Sampling at a sub-Nyquist rate has the benefit of saving in the sampling rate, which can improve power and computational efficiency of the system \cite{timestampless, Walden}. Furthermore, choosing the sampling locations (\emph{nonuniform sampling}), one can further reduce the sampling rate.

Nonuniform sampling works by measuring signal values at a set of arbitrary locations in time. This can occur when we are either unable to uniformly sample the signal due to some physical constraints, or when we lose some of the samples after uniform sampling.  Furthermore, nonuniform sampling is helpful in dealing with aliasing \cite{Bretthorst}. Similar to uniform sampling, it has been shown that for reliable recovery, the average sampling rate must be at least twice the bandwidth of the signal \cite{Landau}.  See \cite{Marvasti} for an overview of nonuniform sampling.

Consider the problem of finding the best locations for sampling a continuous signal to minimize the reconstruction distortion. To find the optimal sampling points, one should utilize any available prior information about the structure of the signal. Furthermore, sampling noise should be taken into account. Due to its importance, nonuniform sampling and its stability analysis has been the subject of numerous  works, in particular for deterministic signals (for instance, see \cite{Lin}). However, there has been relatively less work for a fully Bayesian model of the signal and sampling noise, where one is interested in the statistical \emph{average} distortion of signal reconstruction over a class of signals. Some previous works along this line address the problem of finding the best locations of nonuniform sampling for minimizing reconstruction distortion. These works can be categorized according to their signal and sampling models. None of these works consider sampling of \emph{periodic bandlimited signals}, which is the topic of this paper. A summary of these works and their comparison with our work is provided later in the introduction. 

\textbf{System Model:} Unlike the previous works that consider aperiodic Gaussian stationary sources with a given autocorrelation function, or aperiodic Markov sources, herein we consider a \emph{stochastic continuous periodic signal} with period $T$. Any such signal is characterized by its values in one period  $[0,T]$. Conversely, from a practical perspective, if we are interested in a signal that is defined \emph{only} on a finite interval $[0,T]$, we can periodically extend it and view it as a  periodic signal with period $T$. Similarly, a finite set of $M$ nonuniform samples of the signal in $[0,T]$ can be periodically extended to generate a set of periodic nonuniform samples.

In this paper, we further assume the signal to be bandlimited, with at most $2(N_2-N_1+1)$ non-zero Fourier series coefficients from the frequency $N_1\omega_0$ to $N_2\omega_0$, where $\omega_0=2\pi/{T}$ is the fundamental frequency and $T$ is the signal period, \emph{i.e.,}
\begin{align}S(t)=\sum_{\ell=N_1}^{N_2}[A_\ell\cos(\ell\omega_0 t)+B_\ell\sin(\ell\omega_0 t)], \qquad t\in[0,T].\label{eqdefs22}
\end{align}
There are $2N\triangleq2(N_2-N_1+1) $ free variables $A_\ell$ and $B_\ell$ for $\ell\in[N_1:N_2]$. To obtain information about the signal, we can assume that all of the samples are taken at time instances $t_i$ in the interval $[0,T]$.  A noiseless sample at time $t_i\in [0,T]$, \emph{i.e.}, 
$$S(t_i)=\sum_{\ell=N_1}^{N_2}[A_\ell\cos(\ell\omega_0 t_i)+B_\ell\sin(\ell\omega_0 t_i)]$$
is a linear equation of coefficients $A_\ell$ and $B_\ell$ for $\ell\in[N_1:N_2]$. The $2N$ free variables $A_\ell$ and $B_\ell$ can be uniquely reconstructed from $M=2N$ linearly independent equations. But we consider taking $M$ \underline{\emph{noisy}} samples $S(t_i)+Z_i$ (at time instances $t_i$ that we choose) with sampling noise $Z_i$. The sampling noise can model quantization noise of an A/D convertor, or the noise incurred by transmitting the samples to a fusion center over a communication channel. Perfect reconstruction from noisy samples is not in general feasible, and distortion is inevitable. To minimize the distortion given a maximum number of $M$ samples, we optimize over sampling locations $t_i$ as well as a pre-sampling filter that we consider in our model. See Fig.~\ref{Fig1} for a description of our problem. 

The Fourier series coefficients ($A_\ell$ and $B_\ell$ for $\ell\in[N_1:N_2]$) and the sampling noise $Z_i$ are all assumed to be normal random variables. Therefore, both the input signal $S(t)$ and its reconstruction $\hat{S}(t)$  are random signals, and the reconstruction distortion
$$\text{Sampling Distortion}=\frac 1T\int_{0}^T|S(t)-\hat{S}(t)|^2dt,$$ 
 is also a  random variable. We wish to minimize the \text{Sampling Distortion}, which is a random variable.  To minimize the values that a random variable may take, besides minimizing its mean, we also wish to  quantify how far the distribution of the random variable is spread out around its mean.
The literature in sampling  theory only looks at minimizing the \emph{expected value} of the \text{Sampling Distortion}. However, we are considering \emph{variance} of the \text{Sampling Distortion} as a measure of the \emph{degree of concentration} of the \text{Sampling Distortion} around its mean (as variance is a measure of how far random realizations are spread out from their mean). One of the contributions of this paper is to raise the possibility that minimizing the expected value of error may potentially increase its variance. In such cases, just minimizing MMSE may not be practically appealing if it leads to large error variance.\footnote{
It is possible to conceive other measures that quantify the degree of concentration of reconstruction error around its mean via for instance higher order statistics (or other functions of the pdf). However, the focus of our work is to pick the common measure of variance to {make the point} about the importance of just minimizing the  expected value of error vs. also looking at its concentration around mean.}

  Herein, we consider minimizing both the expected value and variance of the \text{Sampling Distortion}, by choosing the best sampling points and the pre-sampling filter. We denote the minimum of the expected value and variance of the \text{Sampling Distortion} by $\mathsf{D}_{\min}$ and $\mathsf{V}_{\min}$, respectively. A small $\mathsf{D}_{\min}$ guarantees a good \emph{average} performance over all instances of the random signal, while a small $\mathsf{V}_{\min}$ guarantees that with high probability, we are close to the promised average distortion for a given random signal (see also Appendix \ref{AppendixA}). In essence, our goal  is to find the tradeoffs between the number of samples ($M$), the pre-sampling filter, the variance of noise ($\sigma^2$) and the optimal expected value and variance of distortion ($\mathsf{D}_{\min}$ and $\mathsf{V}_{\min}$).

\textbf{Overview of Main Results:}  To state the main results, let us make the following definitions based on the signal description given in equation \eqref{eqdefs22}:
\begin{itemize}
\item We call $(N_2-N_1+1)f_0=Nf_0$ the \emph{signal bandwidth} (of the bandlimited signal), where $f_0=1/T$. 
\item We call $2N_2f_0$ the \emph{Nyquist rate} (twice the maximum frequency of the signal). 
\item We call $M/T=Mf_0$ the \emph{sampling rate}. It is the number of total samples in a period $[0,T]$, divided by  period length $T$. If we periodically extend the $M$ samples (periodic nonuniform sampling), $M/T$ will be the number of samples taken \underline{per unit time}, hence called the sampling rate. 
\end{itemize}

We provide tight results or bounds on the tradeoffs among various parameters such as distortion, sampling rate and sampling noise. 
We first show that to achieve the optimal average and variance of distortion, no pre-sampling filter is necessary for any arbitrary number of samples $M$. Next, when the sampling rate ($Mf_0$) is below the signal bandwidth ($Nf_0$), \emph{i.e.,} $M\leq N$, we find the optimal average and variance of distortion, denoted by $\mathsf{D}_{\min}$ and $\mathsf{V}_{\min}$, respectively. Interestingly, we show that the minima of both $\mathsf{D}_{\min}$ and $\mathsf{V}_{\min}$ are obtained at the same sampling locations. Minimum occurs if we sample at any arbitrary subset of size $M$ of $\{0, T/N, 2T/N, 3T/N, \cdots, (N-1)T/N\}$. If $M<N$, this forms a \emph{nonuniform sampling} set. It is worth to note that the sampling locations only depend on the bandwidth of the signal; they are  optimal for all values of the noise variance ($\sigma^2$) and $N_1$. Moreover, the sampling points are robust with respect to missing samples. Note that the optimal sampling points are any arbitrary $M$ points from the set $\{0, T/N, 2T/N, 3T/N, \cdots, (N-1)T/N\}$. Thus, if we sample at these positions and we miss some samples, \emph{i.e.,}  getting $M'<M$ samples instead of $M$ samples, the set of $M'$ sampling points is still a subset of  $\{0, T/N, 2T/N, 3T/N, \cdots, (N-1)T/N\}$, and hence optimal. Finally, complete characterization of distortion in terms of $M$ and $\sigma^2,$  allows us to answer the question of  whether it is better  to collect a few accurate samples from a signal, or collect many inaccurate samples from the same signal (a problem related to selecting an appropriate $\Sigma\Delta$ modulator). 

For sampling rates ($Mf_0$) above the signal bandwidth ($Nf_0$), we provide  lower and upper bounds that are shown to be tight in some cases. When $Nf_0<Mf_0\leq 2Nf_0$, \emph{i.e.,} we find optimal average distortion when $N$ divides $2N_1-1$, \emph{i.e.,} $N|2N_1-1$, using a non-uniform set of sampling locations. In addition, if sampling rate ($Mf_0$)  is above the Nyquist rate ($2N_2f_0$), \emph{i.e.,} $M>2N_2$, the uniform sampling  is shown to be optimal under certain constraints. Whenever we find $\mathsf{D}_{\min}$ and $\mathsf{V}_{\min}$ explicitly, the  minima are achieved simultaneously at the same optimal sampling points.

We also consider extensions of the results, for example to discrete or sparse periodic signals.

\textbf{Proof technique:} Due to the assumption of Gaussian distribution on the signal coefficients $A_\ell$ and $B_\ell$, minimizing the MMSE of the signal reduces to minimizing the linear MMSE (LMMSE). LMMSE can be expressed as a linear algebra optimization problem over matrices. However, this is a non-linear optimization problem because the coordinates of the matrices include sine and cosine functions of the sampling locations $t_i$ (the variables we are optimizing over are sampling locations $t_{i}$). Sine and cosine functions are nonlinear, albeit structured, functions; their structure can be exploited to solve the optimization problem. A key tool that we repeatedly exploit is an inequality in \emph{majorization theory} that relates trace of a function of a matrix to the diagonal entries of the matrix (see \cite[Chapter 2]{Bhatia} for an overview of majorization theory).

\textbf{Related works:} As mentioned above, the signal model considered in previous works differs from our model and hence it is not possible to directly compare our results with the existing ones. Nonetheless, we remark on similarities and differences among  the conclusions of our work and the other papers. A number of papers consider the reconstruction distortion of a \emph{stationary process} from its noisy  samples  under the \underline{uniform sampling} strategy. Reconstruction of a \emph{lowpass} stationary process from its uniform samples was expressed in an early work \cite{Balakrishnan} in terms of  the samples of the auto-correlation function of the process. The same paper  assumes no  sampling noise, and shows that uniform sampling is sufficient above the Nyquist rate for perfect signal reconstruction. Assuming noisy uniform samples, optimal pre-sampling filter and the corresponding distortion has been found in  \cite{eldarq, Matthews, Chan} via different proof techniques.

Reconstruction distortion from noisy \underline{nonuniform samples} has also been considered in the literature. Authors in \cite{WuSun} consider a \emph{stationary Gaussian signal model} with auto-correlation function $R(\tau)=\rho^{|\tau|}$. They adopt a non-adaptive sampling strategy  and show that amongst all nonuniform sampling strategies, uniform sampling is optimal to minimize the average sampling distortion. This is in contrast with our results on the benefit of nonuniform sampling in the low sampling rate region. However, it relates to our result in the high sampling region where   uniform sampling strategy is optimal. Thus, optimality of uniform sampling strategy depends on the class of stochastic signals one considers. The authors in \cite{timestampless} assume a \emph{first order Markov} signal. They consider an \emph{adaptive sampling} method, wherein the location of the next sample is chosen based on the previous sampling locations and values. The paper employs dynamic programming and greedy techniques to solve this problem. This paper shows that adaptive nonuniform sampling outperforms uniform sampling. Finally, there are some papers such as \cite{app1}-\cite{app4} that consider specialized signal models (e.g. based on differential equations) to study sampling of a variable of interest in a given practical application. There are also less related works (such as \cite{Prakash, Davis, Gastpar}) that discuss the tradeoffs between sampling rate and reconstruction distortion in the context of quantization or compressed sensing.

\textbf{Organization of the Paper:} In Section \ref{sec:problemdef},  the problem is formally defined. Section \ref{sec:overviewmainproofs} provides an overview of the proof techniques used to prove the main results. In Section \ref{sec:main-results},  the main results of the paper are presented; the proofs are relegated to  Appendix \ref{proofs}. In Section \ref{sec:extensions1}, the extensions of our results to the  sparse and discrete signals are considered. Finally the proofs and technical details are given in the appendices.

\begin{table*}
\begin{center}
\begin{footnotesize}
\begin{tabular}{|c|c|}
\hline
Notation & Description \\
\hline
$T$, $f_0$, $\omega_0$ & $T$ is signal period, $f_0=1/T$ and $\omega_0=2\pi f_0$ \\
\hline
$A_\ell, B_\ell$& Fourier series coefficients of the original signal\\ &$\ell\in [N_1:N_2]$\\
\hline
$\mathsf{p}$ & Average signal power in each frequency\\
\hline
 & Support of the input signal in frequency
\\$N$, $N_1$ and $N_2$& domain is from $N_1\omega_0$ to $N_2\omega_0$. 
\\&$N=N_2-N_1+1$.  \\&$2N$ is the number of free variables \\
\hline
 $(N_2-N_1+1)f_0=Nf_0$ & Signal bandwidth \\
$2N_2f_0$ &Nyquist rate \\
\hline
$M$ & Number of noisy samples \\
$M/T=Mf_0$ & Samping rate\\
\hline
$\sigma^2$ & Variance of the sampling noise \\
\hline
$SNR$ & $N\mathsf{p}/\sigma^2$\\
\hline
$\dagger$ & Conjugate transpose \\
\hline
$\{t_1, t_2, \cdots, t_M\}$ & Sampling time instances \\
\hline
$H(\omega)$ & Pre-sampling filter \\
\hline
&  Original signal, the signal after passing the filter\\$S(t)$, $\tilde{S}(t)$ and $\hat{S}(t)$&and its reconstruction, respectively.\\
\hline
$\Delta_1$ & $\ps Q Q^\dagger+\sigma^2 I$\\
\hline
$\Delta_2$ & $\ps Q^\dagger Q+\sigma^2 I$ \\
\hline
\end{tabular}
\end{footnotesize}
\end{center}
\caption{Summary of notations.}
\label{table:summary}
\end{table*}

\section{Problem Definition}\label{sec:problemdef}
\begin{figure*}[t]
\centering
\includegraphics[scale=0.4,angle=0]{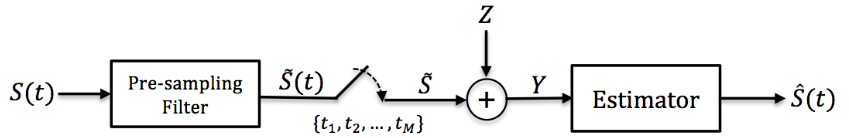}
\caption{\small{Sampling-Distortion Model. The signal $S(t)$ is passed through the pre-sampling filter. The output $\tilde{S}(t)$ is sampled at times $t_1, t_2, \cdots, t_M$. A noise $Z$ is added to the signal at the time of sampling. The noisy observations are then used to recover the signal.}}
\vspace{-0.2cm}
\label{Fig1}
\end{figure*}

We consider a continuous  bandlimited periodic signal defined as follows: \begin{align}S(t)=\sum_{\ell=N_1}^{N_2}[A_\ell\cos(\ell\omega_0 t)+B_\ell\sin(\ell\omega_0 t)],\qquad t\in[0,T]\label{eqdefs}
\end{align}
where $\omega_0=2\pi/{T}$ is the fundamental frequency.  The summation is from $\ell=N_1$ to $N_2$, indicating that the signal is bandlimited. We assume that $A_\ell$ and $B_\ell$ for $N_1\leq \ell\leq N_2$ are mutually independent Gaussian r.v.s  \footnote{
In this paper, we employ reconstruction and distortion formulas for linear MMSE. For Gaussian  distributions, MMSE and linear MMSE are identical.
}, distributed according to $\mathcal{N}(0,\ps)$ for some $\ps>0$.\footnote{
Observe that the term
$A_\ell\cos(\ell\omega_0 t)+B_\ell\sin(\ell\omega_0 t)$ can be expressed as 
$C_\ell \cos(\ell\omega_0 t+\phi_\ell)$,
where $C_\ell=\sqrt{A_\ell^2+B_\ell^2}$. Since  $A_\ell$ and $B_\ell$ are mutually independent Gaussian r.v.s with the same variance, $C_\ell$ has a Rayleigh distribution with zero mean, and $\phi_\ell$ has a uniform distribution. Furthermore, $C_\ell$ and $\phi_\ell$ are mutually independent. Hence, the signal is wide-sense stationary. Since the r.v.s are jointly Gaussian, it is  also strict-sense stationary. 
}
 Thus, the signal power is $N\ps$, where $N = N_2-N_1+1$.

Our model is depicted in Figure \ref{Fig1}. The signal ${S}(t)$, given in  \eqref{eqdefs}, is passed through a pre-sampling filter, $H(\omega)$, to produce $\tilde{S}(t)$. The signal $\tilde{S}(t)$ is sampled at time instances $t_1, t_2, \cdots, t_M$, where $t_i\in [0,T)$.  These samples are then corrupted by $Z(t)$, an i.i.d.\ zero-mean Gaussian noise $\mathcal{N}(0, \sigma^2)$. Thus, our observations are $Y_i=\tilde{S}(t_i)+Z_i, i = 1, 2, \cdots, M$.
An estimator uses the noisy samples $(Y_1,Y_2,\cdots, Y_M)$ to reconstruct the original signal, denoted by  $\hat{S}(t)$. The incurred sampling distortion given by MMSE criterion is equal to
$$\text{Sampling Distortion}=\frac{1}{T}\int_{t=0}^T|\hat{S}(t)-S(t)|^2dt.$$
This distortion is a random variable. Our goal is to compute the minima of the expected value and variance of this random variable, \emph{i.e.,}  to minimize
\begin{align}\mathsf{D}&=\mathbb{E}\left\{\frac{1}{T}\int_{t=0}^T|\hat{S}(t)-S(t)|^2dt\right\},
\\\mathsf{V}&=\mathsf{Var}\left\{\frac{1}{T}\int_{t=0}^T|\hat{S}(t)-S(t)|^2dt\right\}.\end{align}
We are free to choose $H(\omega)$ and sampling locations $t_i$, $ i\in\{1, 2, \cdots, M\}$. Therefore, the optimal average and variance of distortion are defined as follows:
\begin{align}\mathsf{D}_{\min}&=\min_{H(\cdot), t_1, t_2, t_3, \cdots, t_M}\mathbb{E}\left\{\frac{1}{T}\int_{t=0}^T|\hat{S}(t)-S(t)|^2dt\right\},\\
\mathsf{V}_{\min}&=\min_{H(\cdot), t_1, t_2, t_3, \cdots, t_M}\mathsf{Var}\left\{\frac{1}{T}\int_{t=0}^T|\hat{S}(t)-S(t)|^2dt\right\},
\end{align}
where the minimization is taken over the sampling locations and the  pre-sampling filter. 

We assume that $H(\omega)$ is a real LTI filter such that 
\begin{align}
|H(\ell\omega_0)|^2\leq 1 ~~~  for ~all ~~~N_1\leq \ell\leq N_2,\label{al}
\end{align} 
meaning that the frequency gain of the filter is at most one; in other words, we assume that the filter is passive and hence cannot increase the signal energy in each frequency. In particular, all-pass filters  $|H(\omega)|=1$ satisfy this assumption.\footnote{The reason for introducing this assumption is that we can always normalize the power gain of the filter by scaling the sampling noise power. More specifically, replacing some given $H(\omega)$ with $\alpha H(\omega)$ for some $\alpha$, changes $\tilde{S}(t_i)$ to $\alpha \tilde{S}(t_i)$. Hence, $Y_i=\tilde{S}(t_i)+Z_i$ would change to $Y_i=\alpha\tilde{S}(t_i)+Z_i$. But reconstruction from $Y_i, (i=1, 2, \dots, M)$ yields the same distortion as the  reconstruction from $Y_i/\alpha$,  which is  equivalent to dividing the noise power by $\alpha$.}

\section{
Overview of Proof Techniques}\label{sec:overviewmainproofs}
In this section, we provide some  intuitions about the proofs of our main results; the formal proofs  are given in the  appendices. To reconstruct the original signal $S(t)$ given in \eqref{eqdefs}, one should estimate  its Fourier coefficients  $A_\ell$ and $B_\ell$ ($\ell = N_1,N_1+1, \cdots, N_2$).
Let $\mathbf{X}$ be the vector of coefficients $A_\ell$ and $B_\ell$ of the  signal $S(t)$ , \emph{i.e.,}  
\begin{align}\mathbf{X}=[A_{N_1}, \cdots, A_{N_2}, B_{N_1}, \cdots, B_{N_2}]^\dagger,\label{eqn:def:coef:ABprime}\end{align}
where $\dagger$ is the transpose operator. 

The signal ${S}(t)$ is passed through an LTI filter and then sampled at time instances $t_i$ for $i = 1,2 \cdots, M$. Since all the operations are linear, the vector of  $M$ observations can be written as $$\mathbf{Y}=QL\mathbf{X}+\mathbf{Z},$$
where $\mathbf{Z}$ is the noise vector, $L$ is a matrix which represents the LTI filter, and $Q$ is a matrix with sine and cosine coordinates as follows (a formal derivation of the formula is given in Section \ref{sec:matrix-form}):
\footnotesize
$$Q=\begin{pmatrix}\cos(N_1\omega_0 t_1) & \cos((N_1+1)\omega_0 t_1) &\cdots&\cos(N_2\omega_0 t_1)
&\sin(N_1\omega_0 t_1) & \sin((N_1+1)\omega_0 t_1) &\cdots&\sin(N_2\omega_0 t_1)
\\
\cos(N_1\omega_0 t_2) & \cos((N_1+1)\omega_0 t_2) &\cdots&\cos(N_2\omega_0 t_2)
&\sin(N_1\omega_0 t_2) & \sin((N_1+1)\omega_0 t_2) &\cdots&\sin(N_2\omega_0 t_2)\\
&&&\vdots&\vdots&&&\\
\cos(N_1\omega_0 t_M) & \cos((N_1+1)\omega_0 t_M) &\cdots&\cos(N_2\omega_0 t_M)
&\sin(N_1\omega_0 t_M) & \sin((N_1+1)\omega_0 t_M) &\cdots&\sin(N_2\omega_0 t_M)
\end{pmatrix}.$$\normalsize
Estimation of $S(t)$ is equivalent to the estimation of its Fourier coefficients $A_\ell$ and $B_\ell$. If we denote the coefficients of the reconstruction signal by 
$$\hat{\mathbf{X}}=[\hat{A}_{N_1}, \cdots, \hat{A}_{N_2}, \hat{B}_{N_1}, \cdots, \hat{B}_{N_2}]^\dagger,$$ using the Parseval's theorem, we will have
\begin{align}\frac{1}{T}\int_{t=0}^T|\hat{S}(t)-S(t)|^2dt &= \frac12 \sum_{\ell=N_1}^{N_2}\left(|\hat{B}_\ell-B_\ell|^2+|\hat{A}_\ell-A_\ell|^2\right)=\frac 12\|\mathbf{X}-\hat{\mathbf{X}}\|^2\label{eqn:Parsevalprime}.\end{align}

Our goal is to use $\mathbf{Y}$ to estimate the signal with minimal distortion, \emph{i.e.,}  from  
\eqref{eqn:Parsevalprime}, we would like to minimize 
\begin{align}\mathsf{D}=\frac{1}{T}\int_{t=0}^T\mathbb{E}\{|\hat{S}(t)-S(t)|^2\}dt=\frac 12 \mathbb{E}\|\mathbf{X}-\hat{\mathbf{X}}\|^2.\label{eqn:Dvalue1prime}\end{align}
 Since all the random variables are Gaussian, the linear MMSE is optimal and thus we would like to use $QL\mathbf{X}+\mathbf{Z}$ to find $\hat{\mathbf{X}}$ such that $E\|\mathbf{X}-\hat{\mathbf{X}}\|^2$ is minimized. There are two  ways to express the mean square error; see \cite[p.~90]{Luenberger}. These two formulas reduce to the following (as formally shown in  Section \ref{sec:matrix-form}):
\begin{align}
\mathbb{E}\| \mathbf{X}-\hat{\mathbf{X}} \|^2=&(2N-M)\ps +\ps\sigma^2\tr[(\ps QLL^\dagger Q^\dagger + \sigma^2 I)^{-1}]\label{eq11}
\\=&\ps\sigma^2 \tr[(\ps L^\dagger Q^\dagger QL + \sigma^2 I)^{-1}].\label{eq12}
\end{align}

We first minimize the above expressions over the matrix $L$ for a given matrix $Q$. Next, we minimize over the matrix $Q$. The  difficulty in the second step is that we need to minimize the trace of a function of the matrix $Q$ which has coordinates  that depend on the sampling times $t_i$ in a non-linear way.

\emph{Minimizing over the matrix $L$:} The coordinates of  matrix $L$ are determined by the LTI filter. It is shown that the passive filter assumption reduces to the constraint $LL^\dagger \leq I$, where by $A\leq B$ we mean that $B-A$ is positive semi-definite. Hence, $Q(LL^\dagger-I) Q^\dagger \leq 0$, implying that
\begin{align}\ps QLL^\dagger Q^\dagger + \sigma^2 I\leq \ps QQ^\dagger+ \sigma^2 I.\label{eqn:ordernewsq}\end{align}
Furthermore, we know that for any two symmetric positive definite matrices $A$ and $B$,  the relation $A \leq  B$ implies that $B^{-1} \leq  A^{-1}$; this is because the function  $f(t)=-t^{-1}$  is operator monotone \cite[Lemma 2.7]{Eric}. From this fact, we obtain
\begin{align}\tr (\ps QLL^\dagger Q^\dagger + \sigma^2 I)^{-1} \geq \tr (\ps QQ^\dagger + \sigma^2 I)^{-1},\end{align}
meaning that $L=I$ is an optimal choice. Then, from the fact that not using any filter on the signal bandwidth is equivalent with $L=I$, we conclude that pre-filtering is not helpful in reducing the MMSE.

\emph{Minimizing over the matrix $Q$:} When we do not use a pre-sampling filter, the two  MMSE formulas given in \eqref{eq11} and \eqref{eq12} reduce to 
\begin{align}
\mathbb{E}\| \mathbf{X}-\hat{\mathbf{X}} \|^2=&(2N-M)\ps +\ps\sigma^2\tr[(\ps QQ^\dagger + \sigma^2 I)^{-1}]\label{eqn:mmQ1}
\\=&\ps\sigma^2 \tr[(\ps Q^\dagger Q + \sigma^2 I)^{-1}].\label{eqn:mmQ2}
\end{align}
Next, we use the following inequality from majorization theory: Let $\Delta$ be a  real symmetric matrix, and $\Delta_{diag}$ be a diagonal matrix constructed by keeping the diagonal entries of $\Delta$ and setting the off-diagonal entries to be zero, \emph{i.e.,} the $(i,j)$th entry of $\Delta_{\diag}$ is zero if $i\neq j$, and the $(i,i)$th entry of $\Delta_{\diag}$ is the $(i,i)$th entry of $\Delta$. Then we have the following inequality:
$$\tr(\Delta^{-1})\geq \tr(\Delta_{\diag}^{-1}).$$
Alternatively, because $\Delta_{\diag}$ is a diagonal matrix, the inequality says that
\begin{align}\tr(\Delta^{-1})\geq \sum_{i}\frac{1}{\Delta_{(i,i)}}.\label{eqn:traceinequalitya1}\end{align}
Furthermore, the above inequality is strict if $\Delta$ has at least one non-zero off-diagonal entry. In other words, equality holds if and only if $\Delta$ is a diagonal matrix ($\Delta_{\diag}=\Delta$).
In fact, more generally from Exercise II.1.12 and Theorem II.3.1  of \cite{Bhatia}, we have the following theorem:
\vspace{-0.3cm}
\begin{bbox}
\vspace{-0.3cm}
\begin{theorem}\label{thm:lemma:diag}
For any convex function $f(\cdot)$ and any real symmetric matrix $\Delta$, we have
$$\tr(f(\Delta))\geq \tr(f(\Delta_{\diag}))=\sum_{i}f(\Delta_{(i,i)}),$$
where a function of a matrix, $f(\Delta)$, is defined as follows: for a diagonal matrix $\Lambda$, $f(\Lambda)$ is defined by applying $f(\cdot)$ to all the diagonal entries. For a  symmetric real matrix $\Delta$, $f(\Delta)$ is defined as  $Pf(\Lambda)P^{-1}$ if $\Delta=P\Lambda P^{-1}$ for some diagonal matrix $\Lambda$. 

If $f(\cdot)$ is strictly convex, \emph{i.e.,} $f''(x)>0$, then equality in the above relation holds if and only if $\Delta$ is a diagonal matrix.
\end{theorem}
\vspace{-0.3cm}
\end{bbox}
\vspace{-0.3cm}
Now, by setting $\Delta_1=\ps QQ^\dagger + \sigma^2 I$ and using the above trace inequality \eqref{eqn:traceinequalitya1} in equation \eqref{eqn:mmQ1}, we get a lower bound on the reconstruction distortion:
\begin{align}
\mathbb{E}\| \mathbf{X}-\hat{\mathbf{X}} \|^2=&(2N-M)\ps +\ps\sigma^2\tr\left(\Delta_1^{-1}\right)\label{eqn:touselateram1}
\\\geq& (2N-M)\ps +\ps\sigma^2\sum_{i}\left(\Delta_{1(i,i)}\right)^{-1}.
\end{align}
The above formula leads to what we call the \emph{Lower bound 1} on distortion. 
Similarly, by setting $\Delta_2=\ps Q^\dagger Q + \sigma^2 I$ and using the above trace inequality \eqref{eqn:traceinequalitya1} in equation \eqref{eqn:mmQ2}, we get another lower bound on the reconstruction distortion. This leads to what we call the  \emph{Lower bound 2} on distortion. 
These two  general lower bounds are formally stated as Lemma \ref{lemma:lower:1}  and Lemma \ref{lemma:lower:2} later. 

The next question is whether any of these two lower bounds can be tight. The trace inequality \eqref{eqn:traceinequalitya1} is tight if and only if $\Delta$ is diagonal. Then, the question is whether we can choose the values of sampling times $t_i$ such that either $\Delta_1=\ps QQ^\dagger + \sigma^2 I$ or $\Delta_2=\ps Q^\dagger Q + \sigma^2 I$ become diagonal matrices. Each off-diagonal entry of $\Delta_1$ or $\Delta_2$ is a function of the sampling times $t_i$. Setting an off-diagonal entry to zero imposes an equation on $t_i$. The question is whether we can simultaneously solve all of these (non-linear) equations corresponding to all the off-diagonal entries. It turns out that if the number of samples $M$ is small (less than or equal to $N$), then we appropriately choose $t_i$ (in a nonuniform manner in the $[0,T]$ interval) such that $\Delta_1$ becomes diagonal (Theorem \ref{T1a}). If $M$ is large (greater than $2N$), then under some constraints we can make $\Delta_2$ diagonal; when $M$ is very large (greater than $2N_2$), the uniform sampling strategy always makes $\Delta_2$ diagonal (Theorem \ref{T1c}). Unfortunately, when $M$ is in the intermediate range, between $N$ and $2N$, we cannot make either $\Delta_1$ or $\Delta_2$  become diagonal with appropriate choice of $t_i$. Nonetheless, backed by numerical simulations, one possible strategy is to choose  the $t_i$ in such a way  that forces a good fraction (if not all) of the entries of $\Delta_1$ to become zero. While it is not necessarily true that this  approach for choosing $t_i$ is optimal, nevertheless the reconstruction distortion incurred by it provides an upper bound on the optimal reconstruction distortion (Theorem 5). To find an \emph{explicit} analytical expression for the upper bound, we use some further tools from linear algebra to be  discussed in the Appendix \ref{sec:B81}.

\section{Main Results} \label{sec:main-results}
In this section, we first claim that no pre-sampling filter is necessary to achieve the optimal distortion, and then state our  results on the optimal choices of  the  sampling locations.
Pre-sampling filter can be potentially useful because it allows one to remove some of the frequency components at the beginning, and thereby enable the sampler to recover the remaining components with a lower estimation error. The total error will be equal to the variance  of the filtered components, plus the estimation error of the unfiltered components. It turns out that such a strategy cannot outperform the naive strategy of not using any filter at all. Our first result states this fact:
\begin{theorem}\label{thm:pre-sampling} 
Consider  the class of passive pre-sampling filters, \emph{i.e.,} real filters $H(\omega)$ satisfying
\begin{align}
|H(\ell\omega_0)|^2\leq 1 ~~~  for ~all ~~~N_1\leq \ell\leq N_2.
\end{align}
Then, the  average and variance of the distortion  are minimized by the identity filter $H(\omega)=1$ for all $\omega$ (which is equivalent with not using any pre-sampling filter).
\end{theorem}
The above theorem implies that not using any pre-sampling filter is optimal. In particular, a pre-sampling filter that reduces the signal bandwidth to \emph{half} the sampling rate (an anti-aliasing filter) is suboptimal. In  Section \ref{sec:filter-Half}, we discuss this suboptimality in more details. Besides this section, In the subsequent parts of the paper we mainly avoid discussion of pre-sampling filter.

In the following, we state two general lower bounds on the average and variance of  distortion and then, using these lower bounds, we provide  the main results of the paper. Note that $SNR = N\ps /\sigma^2$, and $N=N_2-N_1+1$ which  is proportional to the bandwidth of the signal (which is $Nf_0$). Remember that $M$ is the number of noisy samples  which is proportional to the sampling rate (defined as $Mf_0$). The  proofs of the main results  are given  in Appendix \ref{proofs}. We did not provide the formal proofs below the statement of the theorems because we first need to represent the problem in a matrix form, as given in Appendix \ref{sec:matrix-form}. 

\subsection{Lower bound 1 on distortion}\label{sec:lemma:1:lower:ab}
As we mentioned in Section \ref{sec:overviewmainproofs}, using the two forms of the MMSE and Theorem \ref{thm:lemma:diag}, one can obtain lower bounds on the optimal average and variance of the reconstruction distortion. We begin by providing the first of these two lower bounds:

\begin{lemma}(Lower bound 1)\label{lemma:lower:1} 
For any noise variance $\sigma>0$, the following lower bounds hold on the minima of the average and variance of the reconstruction distortion for any values of $M$, $N$ and $SNR$:
\begin{align*}\frac{\mathsf{D}_{\min}}{\ps}&\geq\frac 12 \left(2N-M+\frac{M}{1+SNR}\right),\\ \frac{\mathsf{V}_{\min}}{\ps^2}&\geq 2N-M+\frac{M}{(1+SNR)^2}.
\end{align*}
The above inequalities are tight for some values of $M$ and $N$. The condition for equality in the above equations is the possibility of choosing the sampling times $\{t_1, t_2, \dots, t_M\}$ in such a way that 
$t_i\neq t_j$ for $i\neq j$,  and moreover for any $i\neq j$ time instances $t_i$ and $t_j$ satisfy one of the following two equations:
\begin{align}|t_i - t_j|&=T\frac {m_1}{N},& ~ \text{for some natural number } m_1,\label{cond-1-lemma1}\\ 
|t_i - t_j|&=T\frac{2m_2-1}{2(N_1+N_2)},&  \text {for some natural number }  m_2. \label{cond-2-lemma1}
\end{align}
\end{lemma}

 Observe that when $M\leq N$, if we choose $t_i=i(T/N)$ for $i=1,2,\cdots, M$, then $|t_i-t_j|=|i-j|T/N$ is an integer multiple of $T/N$. Hence, the above condition can be satisfied for any $M\leq N$. On the other hand, this condition cannot hold if $M>2(N_1+N_2)$. This is because any set of distinct points $t_1, t_2, \dots, t_{M}\in[0,T]$ will have two points whose difference $|t_i-t_j|$ is less than or equal to $T/M$. Since $M>2(N_1+N_2)$, $|t_i-t_j|$ cannot be an integer multiple of $T/N$ or $T/(2(N_1+N_2))$.

See Appendix \ref{Lower1proof} for a proof. Roughly speaking, the proof uses 
Theorem \ref{thm:lemma:diag}  as discussed in Section \ref{sec:overviewmainproofs} for the matrix $\Delta_1=\ps QQ^\dagger + \sigma^2 I$. It  explicitly works out the condition on sampling times $t_i$ that makes the matrix $\Delta_1$  diagonal.

The above observation leads to the following theorem for $M\leq N$:

\begin{theorem}\label{T1a} 
For $M\leq N$, the optimal average and variance of distortion are 
\begin{align}\frac{\mathsf{D}_{\min}}{\ps}&=\frac 12 \left(2N-M+\frac{M}{1+SNR}\right),\label{eqn:anewT1a}\\
\frac{\mathsf{V}_{\min}}{\ps^2}&=2N-M  + \frac{M}{(1 +SNR)^2}.\label{eqn:anewT1aVar}\end{align}
Both the 
minimal average and variance of the distortion are obtained by choosing $M$ distinct time instances arbitrarily from the following set of  $N$ samples \footnote{ Observe that these points are uniform sampling points with increment $T/N$; they are not uniform sampling points corresponding to $M$ samples, unless $M=N$. Uniform sampling for $M$ points takes the samples at times $\{0, \frac{1}{M}T, \frac{2}{M}T, \frac{3}{M}T, \cdots, \frac{M-1}{M}T\}.$} 
$$\{0, \frac{1}{N}T, \frac{2}{N}T, \frac{3}{N}T, \cdots, \frac{N-1}{N}T\}.$$
 The optimal interpolation formula for this set of sampling points is given by
\begin{align}\hat{S}(t)&=\frac {\ps}{N\ps+\sigma^2} \sum_{i=1}^{M} \sum_{\ell=N_1}^{N_2}\cos(\ell \omega_0 (t-t_i))Y_i, \nonumber
\\&=\frac {\ps}{N\ps+\sigma^2} \sum_{i=1}^{M} \cos(\frac{\omega_0}{2} (t-t_i)) \cdot \frac {\sin(\frac{N\omega_0}{2} (t-t_i))}{\sin(\frac{\omega_0}{2} (t-t_i))} \cdot Y_i,\label{eqn:Interpolation1}\end{align}
where $Y_i$ is the noisy sample of the signal at $t=t_i$.
\end{theorem}

See Appendix \ref{proof:4:2:1} for a proof. To understand the meaning of  \eqref{eqn:anewT1a}, observe that for fixed values of $\ps$ and $\sigma$, the minimum distortion is linear in $M$. But  it is not linear in $N$ (as SNR in the denominator is $N\ps/\sigma^2$), except when $SNR$ goes to infinity. In the case of noiseless samples, $\sigma=0$, the SNR will be infinity and the minimal distortion will be $$\frac{\ps}{2}(2N-M).$$ To intuitively understand this equation, observe that there are $2N$ free variables and we can recover $M$ of them using the samples. Therefore, we will have $2N-M$ free variables, giving a total distortion of $(2N-M)\ps/2$ as the power of each sinusoidal function is $\ps/2$. Moreover, the maximum distortion is $2N(\ps/2)$, which is obtained when 
$\sigma=\infty$  ($SNR=0$) or $M=0$. 
Finally, observe that when $SNR$ is large, 
\begin{align*}\mathsf{D}_{\min}&\approx\frac {\ps}{2} \left(2N-M+\frac{M}{SNR}\right)
=\frac {\ps}{2}(2N-M) +\frac{M\sigma^2}{2N}\end{align*} 
which is linear in both $\ps$ and $\sigma^2$. Observe that $\mathsf{D}_{\min}$ is increasing in both $\ps$ and $\sigma^2$, since $2N-M\geq 0$ for $M\leq N$. Furthermore, $(2N-M)/2\geq M/(2N)$ when $M\leq N$, implying that the coefficient of $\ps$ is greater than or equal to that of $\sigma^2$.  Figure \ref{Fig3} depicts $\mathsf{D}_{\min}$ as a function of $\ps$, $\sigma^2$, $M$ and $N$.

\iffalse Another observation is that the distortion when $M=0$ is equal to $\mathsf{p}N$. When $M=N$, the distortion equals $\mathsf{p}N(1+\frac{1}{1+SNR})/2\geq \mathsf{p}N/2$. As a result, we can conclude that when sampling rate is at most half of the signal bandwidth, even the optimal nonuniform sampling strategy cannot achieve a gain of more than $3dB$ in comparison to  not taking any samples.\fi

\begin{figure*}
\centering
\includegraphics[scale=0.45,angle = 0]{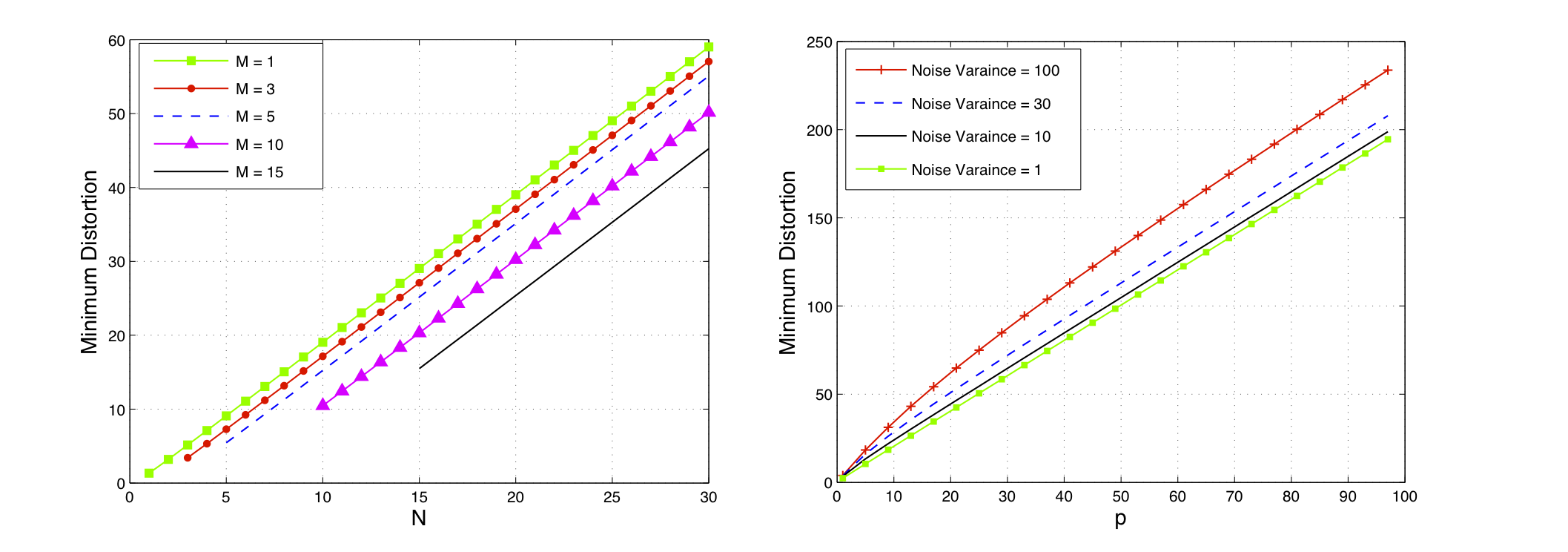}
\caption{{  In the left subfigure, the minimum average distortion is illustrated as a function of   $M$ for different values of $M\leq N$.  In the right subfigure,  the minimum average distortion is depicted as a function of  $\ps$ for different values of the noise variance $\sigma^2$.}}
\label{Fig3}
\end{figure*}

In the statement of Theorem \ref{T1a}, one possible optimal sampling set was given. 
Next, we consider the question of whether other optimal sampling sets exist for $M\leq N$:
\begin{proposition}\label{thm:unique}  If $M\leq N$, and $N$ does not divide $2N_1-1$, then a given set of sampling times $(t_1, t_2, \cdots, t_M)$ is optimal \emph{only if} there exists some constant $\tau<T/N$ such that for at least $M-1$ of the sampling points $t_i$ we have
$$t_i\in \{\tau, \tau+\frac{1}{N}T, \tau+\frac{2}{N}T, \tau+\frac{3}{N}T, \cdots, \tau+\frac{N-1}{N}T\}.$$
\end{proposition}
See Appendix \ref{proof:4:2:3} for a proof.

\subsection{Lower bound 2 on distortion}\label{sec:lemma2:section:abs}

In the next lemma, alternative lower bounds on $\mathsf{D}_{\min} $ and $\mathsf{V}_{\min}$  are given, which are tighter than the ones given in Lemma \ref{lemma:lower:1}  for $M > 2N$.

\begin{lemma}(Lower bound 2)\label{lemma:lower:2}
 For any $\sigma>0$, the following lower bounds hold for all values of $M, N$ and $SNR$:
\begin{align} 
\mathsf{D}_{\min} &\geq \frac{N\ps }{1+\frac {M}{2N} SNR },
\\ \mathsf{V}_{\min}&\geq\frac{2N\ps^2 }{(1+\frac {M}{2N} SNR)^2 }.
\end{align}

The above inequalities are tight for some values of $M$ and $N$. The condition for equality in the above equations is the possibility of choosing distinct sampling times that satisfy the following equation:
\begin{align}
\sum_{i=1}^{M} e^{j2 \pi k\frac{t_i}{T}} = 0 \qquad \emph{for} \quad k\in\{1, 2, \cdots, N-1\}\cup\{2N_1, 2N_1+1, \cdots, 2N_2\}.
\label{eqn:for:thm:unique3inlemma2}
\end{align}
\end{lemma}

See Appendix \ref{Lower2proof} for a proof.  Roughly speaking, the proof uses 
Theorem \ref{thm:lemma:diag}  as discussed in Section \ref{sec:overviewmainproofs} for the matrix $\Delta_2=\ps Q^\dagger Q + \sigma^2 I$. The sampling times deriven from  the above condition diagonalize tha matrix $\Delta_2$. As an example, consider $M>2N_2$ and uniform samples $t_i=(i-1)(T/M)$ for $i=1,\dots, M$. Then, using geometric series calculus, we have 
$$
\sum_{i=1}^{M} e^{j2 \pi k\frac{i-1}{M}} = \frac{e^{j2 \pi k\frac{M}{M}}-1}{e^{j2 \pi k\frac{1}{M}}-1}=0$$
for any natural number $1\leq k\leq M-1$. This shows that the bounds given above are tight via uniform sampling if $M>2N_2$. 

The above observation leads to the following theorem for $M\geq 2N$:

\begin{theorem}\label{T1c} 
The following lower bounds on optimal average and variance of distortion hold:
\begin{align}\frac{\mathsf{D}_{\min}}{\ps} &\geq \frac{N}{1+\frac{M}{2N}SNR},\label{eqn:25thm71}
\\ \frac{\mathsf{V}_{\min}}{\ps^2}&\geq\frac{2N }{(1+\frac {M}{2N} SNR)^2 }.\label{eqn:25thm72}\end{align}
Furthermore, for $M\geq 2N$, uniform sampling, \emph{i.e.,}  $t_i=iT/{M}$ ($ i=1, 2, \cdots, M$) is an optimal sampling strategy if no multiple of $M$ can be found in the set $\{2N_1, 2N_1+1, \cdots, 2N_2\}$. In particular, uniform sampling is optimal when $M>2N_2$.
Also the reconstruction formula for the optimal set of sampling points is given by 
\begin{align}\hat{S}(t)&=\frac {\ps}{\frac{M}{2}\ps+\sigma^2} \sum_{i=1}^{M} \sum_{\ell=N_1}^{N_2}\cos(\ell \omega_0 (t-t_i))Y_i, \nonumber
\\&=\frac {\ps}{N\ps+\sigma^2} \sum_{i=1}^{M} \cos(\frac{\omega_0}{2} (t-t_i)) \cdot \frac {\sin(\frac{N\omega_0}{2} (t-t_i))}{\sin(\frac{\omega_0}{2} (t-t_i))} \cdot Y_i,\label{eqn:Interpolation2}\end{align}
where $Y_i$ is the noisy sample of the signal at $t=t_i$.
\end{theorem}
See Appendix \ref{proof:4:4:1} for a proof. The condition given for the optimality of uniform sampling comes from \eqref{eqn:for:thm:unique3inlemma2} in the statement of Lemma 
\ref{lemma:lower:2}.
To intuitively understand the statement of the theorem, observe that when $M$ goes to infinity, the minimum distortion $\mathsf{D}_{\min}$ goes to zero with order $\bigO{1/M}$, regardless of the value of sampling noise. This is because when $M$ is very large ($M>2N_2$), the lower bounds given in \eqref{eqn:25thm71} and \eqref{eqn:25thm72} are tight. An application for sampling  above the Nyquist rate ($M>2N_2$) is bandpass sampling in the intermediate frequency (IF). 

Unlike the case of $M\leq N$, where the optimal sampling points could be found without any need to know $N_1$, observe that the constraint that no multiple of $M$ can be found in the set $\{2N_1, 2N_1+1, \cdots, 2N_2\}$ depends on $N_1$. In fact, the uniform sampling strategy, $t_i=iT/M$, is not in general optimal. For instance, uniform sampling is not an optimal strategy if there is some natural number $k\in\{1, 2, \cdots, N-1\}\cup\{2N_1, 2N_1+1, \cdots, 2N_2\}$ that is a multiple of $M$, since in this case  \eqref{eqn:for:thm:unique3inlemma2} becomes
$\sum_{i=1}^{M} e^{j2 \pi k\frac{t_i}{T}}=\sum_{i=1}^{M} e^{j2 \pi k\frac{i}{M}}=M \neq 0.$
In fact, non-uniform sampling may be optimal in some cases. Below are some examples of non-uniform sampling sets that are the  solutions of   \eqref{eqn:for:thm:unique3inlemma2}:

\begin{example}\emph{ When $M=2$, $N=1$ and $N_1=1$, sampling points $t_1=0, t_2=T/4$ are optimal.} \end{example}
\begin{example} \emph{Given any $N_1$ and $N_2$, let $\mathcal{K}=\{1, 2, \cdots, N-1\}\cup\{2N_1, 2N_1+1, \cdots, 2N_2\}$. Let us show $\mathcal{K}$ by $\mathcal{K}=\{k_1, \cdots, k_{|\mathcal{K}|}\}$. Then the following $M=2^{|\mathcal{K}|}$ non-uniform sampling points are optimal for $i\in [0:M-1]$. let
$$t_i=\sum_{r=1}^{|\mathcal{K}|} b_r\frac{T}{2k_r},$$
where $(b_1, \cdots, b_{|\mathcal{K}|})\in\{0,1\}^{|\mathcal{K}|}$ is the binary expansion of time index $i\in[0:M-1]$. We verify that \eqref{eqn:for:thm:unique3inlemma2} holds:
\begin{align}\sum_{i} e^{j2 \pi k\frac{t_i}{T}}&=\sum_{b_1, \cdots, b_{|\mathcal{K}|}}e^{j2 \pi k\sum_{r=1}^{|\mathcal{K}|} b_r\frac{1}{2k_r}}
\\&=\prod_{r=1}^{|\mathcal{K}|}\left(\sum_{b_r}e^{j\pi b_r\frac{k}{k_r}}\right)
\\&=\prod_{r=1}^{|\mathcal{K}|}\left(1+e^{j\pi\frac{k}{k_r}}\right),
\end{align}
which is zero if $k\in\mathcal{K}$, \emph{i.e.,}  $k=k_r$ for some $r$.}
\end{example}

\subsection{Sampling Distortion Tradeoffs for $N< M\leq 2N$}\label{sec:4:3}

Given an  arbitrary $M\leq 2N$, Lower Bound 1 (Lemma \ref{lemma:lower:1})   is tight when $N$ divides $2N_1-1$, \emph{i.e.,} $N|2N_1-1$. This is because  $N|2N_1-1$ implies that $N$ divides $2N_1-1+N=N_1+N_2$, and using conditions
\eqref{cond-1-lemma1} and \eqref{cond-2-lemma1} from the statement of Lemma \ref{lemma:lower:1},  any subset of size $M$ of 
\begin{align*}\bigg\{0&, \frac{1}{N}T, \frac{2}{N}T, \cdots, \frac{N-1}{N}T, \\&
\frac{T}{2(N_1+N_2)}, \frac{1}{N}T+\frac{T}{2(N_1+N_2)}, \frac{2}{N}T+\frac{T}{2(N_1+N_2)}, \cdots, \frac{N-1}{N}T+\frac{T}{2(N_1+N_2)}\bigg\},\end{align*}
is an optimal sampling set and thus Lower Bound 1 holds with equality.

However, when $N$ does not divide $2N_1-1$, one cannot choose the sampling times so that the equality conditions given in the two general lower bounds are met for $N<M\leq 2N$. Our general idea, as discussed in Section \ref{sec:overviewmainproofs}, is to choose the sampling times so that either $\Delta_1=\ps QQ^\dagger + \sigma^2 I$ or $\Delta_2=\ps Q^\dagger Q + \sigma^2 I$ become diagonal. In general, this is not possible  when $N<M\leq 2N$. We observed in numerical simulations that the optimal sampling set makes many off-diagonal entries of $\Delta_1$ zero. Therefore, our idea is to look for some sampling strategies that make $\Delta_1$ as diagonal as possible, and then attempt to find \emph{explicit expressions} for the resulting distortion using tools from linear algebra. Any particular sampling strategy would give us an upper bound on $\mathsf{D}_{\min}$. In our sampling strategy, we take an optimal sampling strategy for $M=N$, and append it by adding $M-N$ extra sampling times as follows:
 \begin{align}\{t_1,t_2,\dots, t_M\}=\{0, \frac{1}{N}T, \frac{2}{N}T, \frac{3}{N}T, \cdots, \frac{N-1}{N}T,~~\frac{1}{2N}T, \frac{3}{2N}T,\cdots, \frac{2M-2N-1}{2N}T\}.\label{samplingpointsupper}\end{align}
This leads to the following upper bound on the average the distortion.

\begin{theorem}(Upper bound)\label{T:NM2N} 
For $N< M \leq 2N$, the optimal average distortion can be bounded as follows
\begin{align}\frac{\mathsf{D}_{\min}}{\ps} \leq&\frac 12 (2N-M)+ \frac{2N-M}{2(1+SNR)}+\mathsf{Num}\cdot\frac{ 1+SNR}{1+2~ SNR }+(M-N-\mathsf{Num}) \cdot\frac{1}{ 1+SNR} ,\label{eqn:17s}
\end{align}
  in which $$\mathsf{Num}=\min\big(f(N_1,N), M-N\big)$$
and $f(a,b)$ is equal to
$$f (a,b)= \begin{cases}
   b-1      & \text{if}  ~ r = 0\\
   2b-2r+1       & \text{if } 2r >b \\
   2r-1  & \text{if } 0<2r \leq b 
  \end{cases}, $$
where $r$ is the remainder of dividing $a$ by $b$.
\end{theorem}

See Appendix \ref{proof:4:3} for a proof.  Here we only state some remarks about the proof.  The set of sampling points given in \eqref{samplingpointsupper} has a simple structure.
 As it becomes clear from the proof, this makes $\Delta_1=\ps QQ^\dagger + \sigma^2 I$ having the following block form ($c$ is a constant):
$\Delta_1=c\begin{bmatrix} I&\Gamma\\ \Gamma^{\dagger}&I\end{bmatrix}$
with two identity blocks appearing on the diagonal. Thus, while this choice cannot make $\Delta_1$ diagonal (the off-diagonal block $\Gamma$ has non-zero entries), but it produces two identity matrix on the main block. The distortion depends on the inverse of $\Delta_1$ (see \eqref{eqn:touselateram1}). While it is possible to relate eigenvalues of $\Delta_1$ to singular values of $\Gamma$ via the Jordan--Wielandt theorem (Theorem \ref{L3}),  it is not easy to find an explicit expression for the singular values of $\Gamma$. This is due to the fact that the entries of $\Gamma$ are summations of some sine and cosine terms and do not have a clean expression. Our idea is to expand the matrix $\Gamma$ by adding new rows or columns, such that the singular values of the extended matrix can be explicitly calculated. Next, we use results from linear algebra that relate singular values of a matrix to the singular values of its submatrices. While expression of the upper bound on the distortion looks complicated (due to the fact that various linear algebra tools are invoked), but we emphasize that the given expression is nonetheless \emph{explicit} and can be used for analytical considerations

\subsection{Plots of the lower and upper bounds}
In this section, we provide a number of plots that illustrate the tradeoffs between the number of samples ($M$), the signal bandwidth $(Nf_0)$, the variance of noise ($\sigma^2$) and the optimal expected value and variance of distortion ($\mathsf{D}_{\min}$ and $\mathsf{V}_{\min}$).

In Fig.~\ref{Lower12}, the two lower bounds on the optimal average distortion $\mathsf{D}_{\min}$ derived in Lemmas \ref{lemma:lower:1} and \ref{lemma:lower:2} are plotted as a function of number of samples $M$, where we have fixed all the other parameters such as $N$ and $\sigma^2$. As we see in this figure, Lower bound 1 is larger than Lower bound 2 when $M$ is small. This is consistent with the fact that Lower Bound 1 is tight and equal to $\mathsf{D}_{\min}$ for $M\leq N$, as shown in Theorem \ref{T1a}; and Lower Bound 2 is tight when $M$ is large, for $M>2N_2$, as shown in Theorem \ref{T1c}. 

 %$(N_1, N, \ps, \sigma) = (9, 9, 1, 1)$

The maximum of the two lower bounds (Lower Bound 1 and Lower Bound 2) is also a lower bound on the optimal average distortion $\mathsf{D}_{\min}$. This is plotted with blue stretched-line in  Fig.~\ref{lowersupper2}. This curve is matches $\mathsf{D}_{\min}$ for $M\leq N=9$ and $M> 2N_2=34$. In this figure, the distortion of uniform sampling is  also depicted (purple dotted-line curve), which constitutes an upper bound on the optimal minimum average distortion $\mathsf{D}_{\min}$.  Therefore, the optimal $\mathsf{D}_{\min}$ lies somewhere in between the stretched-line blue curve and the dotted-line purple curve. Remember that for the case of $N=9\leq M\leq 2N=18$, we had an explicit upper bound on the optimal average distortion $\mathsf{D}_{\min}$ in Theorem \ref{T:NM2N}. This upper bound is the red dashed curve. We make the following observations about these curves:
\begin{itemize}
\item The performance of the uniform sampling is close to optimal  for $M\leq N=8$ (its curve almost matches that of the  lower bound in Lemma \ref{lemma:lower:1} which is optimal for $M\leq N$). From here, we conclude that while uniform sampling is not optimal for $M\leq N$, it is near optimal. A numerical calculation shows that we get a percentage gain of about 0.3 by using the optimal non-uniform sampling strategy. However, a main advantage of the optimal nonuniform sampling is its \emph{robustness} with respect to missing samples (as discussed in the introduction).
\item The distortion of uniform sampling is not monotonically decreasing in the sampling frequency. While increasing the number of  samples leads to a better performance in nonuniform sampling, this is not necessarily the case for uniform sampling. 
\item The curve for uniform sampling distortion matches the second lower bound at $M=35$ onwards; uniform sampling is optimal for $M> 2N_2 = 34$ (consistent with  Theorem \ref{T1c}).
\item For the given choice of parameters the explicit upper bound of Theorem \ref{T:NM2N} can come below the distortion of uniform sampling. 
\end{itemize}

\begin{figure*}[t]
\centering
\includegraphics[scale=0.3,angle = 0]{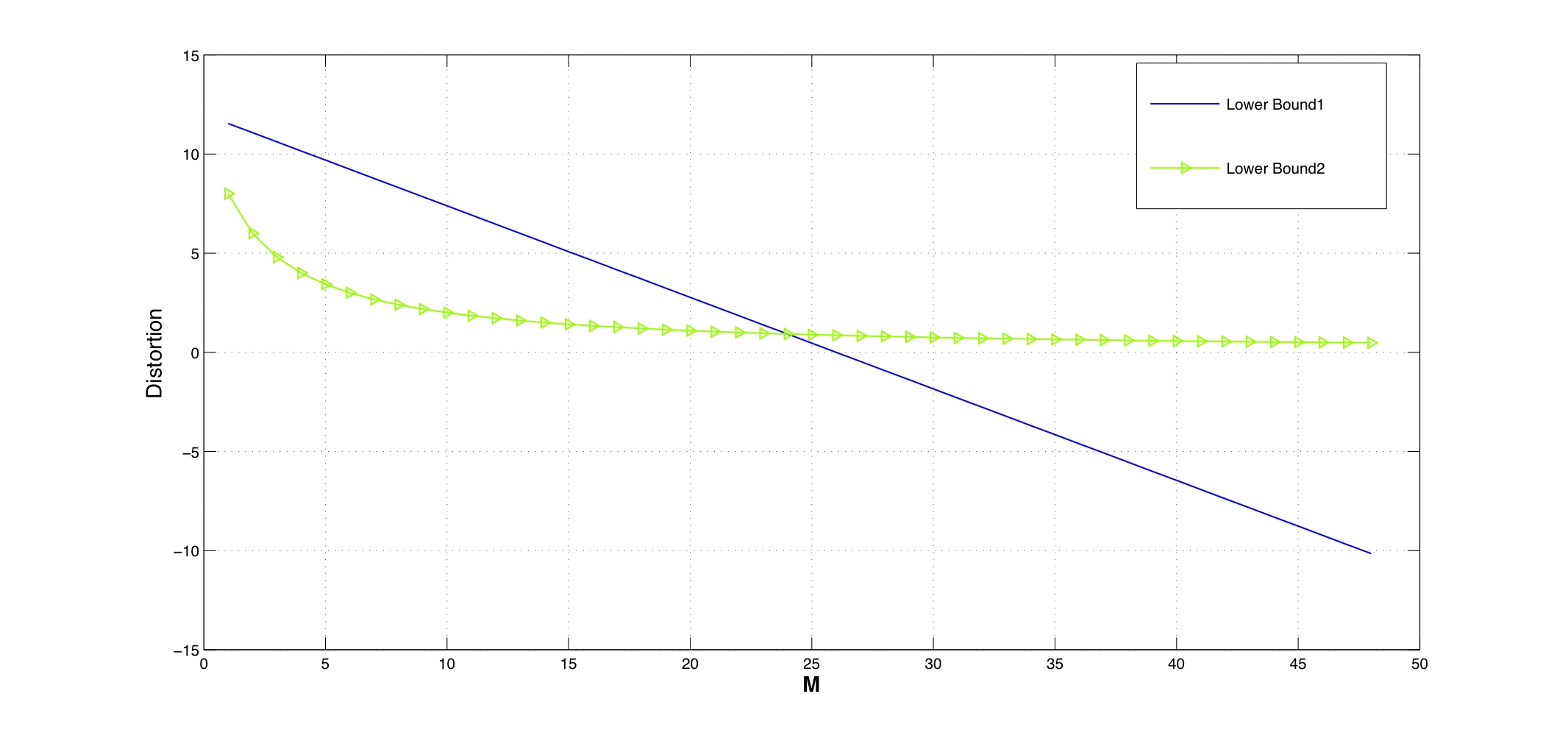}
\caption{{  The lower bounds on the  average  distortion  given in Lemmas \ref{lemma:lower:1} and \ref{lemma:lower:2}.}}
\vspace{0.4cm}
\label{Lower12}
\end{figure*} 

\begin{figure*}[t]
\centering
\includegraphics[scale=0.32,angle=0]{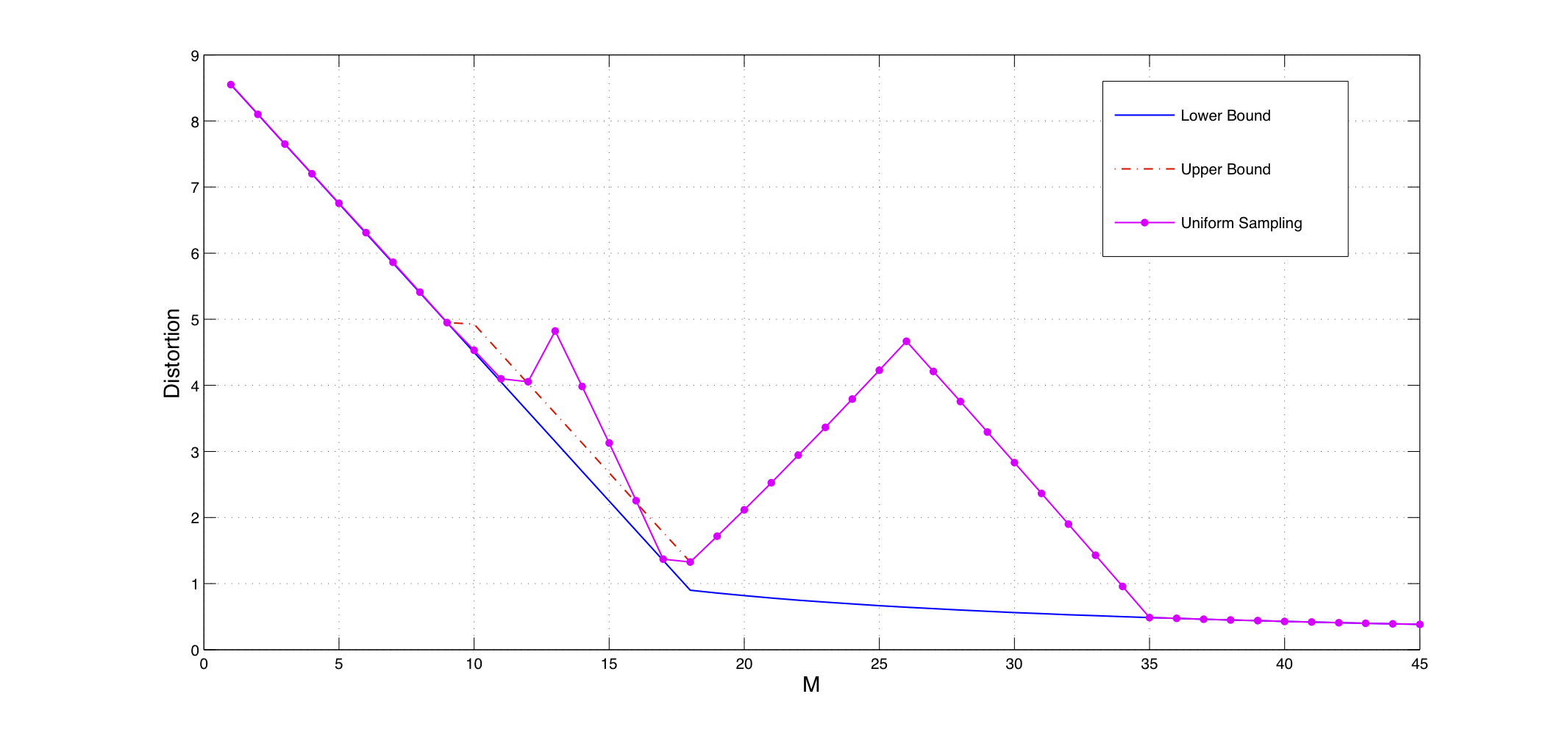}
\caption{{  The average  distortion for $(N_1, N, \ps, \sigma) = (9, 9, 1, 1)$. The signal period is $T=1$. Here $M$ is shown on the x-axis. The union lower bound (blue stretched-line curve) and the upper bound  given in Theorem \ref{T:NM2N} (red dashed curve) for $M$ between $N$ and $2N$ are depicted. The performance of uniform sampling without using a  pre-sampling filter is also plotted (purple dotted-line curve). }}
\label{lowersupper2}
\end{figure*}

\iffalse
\begin{figure*}[t]
\centering
\includegraphics[scale=0.2,angle=0]{Lower3.png}
\caption{{  The difference of the distortions of the uniform  and optimal sampling, divided by the distortion of the optimal sampling. The parameters for this figure are $(N_1, N, \ps, \sigma) = (4, 7, 1, 1)$.}}
\label{lowersupper3}
\end{figure*}
\fi

\subsection{Suboptimality of filtering at half the sampling rate }\label{sec:filter-Half}

Consider the case of $M=N$. According to Theorem \ref{T1a}, the uniform sampling set
\begin{align}\{t_1, t_2, \cdots, t_M\}=\{0, \frac{1}{N}T, \frac{2}{N}T, \frac{3}{N}T, \cdots, \frac{N-1}{N}T\}\label{eqn:anti-aliasing2}\end{align}
is optimal and obtains the minimum distortion
$$\mathsf{D}_{\min}=\frac {\ps}{2} \left(M+\frac{M}{1+SNR}\right).$$
Here, it is assumed that no pre-sampling filter is used  (following from the statement of Theorem \ref{thm:pre-sampling}). Observe that the signal bandwidth, $Nf_0$ is equal to the sampling rate $Mf_0$. Here the sampling rate is \emph{not} twice the signal bandwidth. Suppose we use the uniform sampling points given in \eqref{eqn:anti-aliasing2}, and employ a pre-sampling filter that reduces the signal bandwidth to \emph{half} the sampling rate as follows (an anti-aliasing filter):
\begin{align}H(\ell \omega_0)=\begin{cases}1,&N_1\leq\ell< N_1+M/2
\\0,&\text{otherwise}\end{cases}.\label{eqn:filtereq1}\end{align}
Here the number of samples, $M$, is assumed to be even.  This filter makes all of the coefficients $A_{\ell}$ and $B_\ell$ equal to zero, except for $M$ variables $A_\ell, B_\ell$ when $N_1\leq\ell< M/2+N_1$.  We then have the following theorem:

\begin{theorem} \label{thm:unifMlessN}

Suppose that we insist on using uniform sampling at rate $M$, \emph{i.e.,}  taking the points $$\{t_1, t_2, t_3, \cdots, t_M\} =\{0, \frac{1}{M}T, \frac{2}{M}T, \frac{3}{M}T, \cdots, \frac{M-1}{M}T\}.$$ 
Then, the performance of the filter $H$ satisfies:
\begin{align}\mathsf{D}_{H}>\frac{\ps}{2} \left(M+\frac{M}{1+\frac 1{2}SNR}\right)\label{eqn:dmind2},\end{align}
which is strictly greater than 
$$\mathsf{D}_{\min}=\frac {\ps}{2} \left(M+\frac{M}{1+SNR}\right),$$
even when we increase the signal power by a factor of two.
\end{theorem}
See Appendix \ref{proof:4:2:2} for a proof.
\begin{discussion}\emph{
 %% just consider the best filter , BW = M and think of using the filter when M>N,
Anti-aliasing filter $H$ is analogous to an \emph{interference cancellation} scheme when we  interpret aliasing as interference. But without a filter, we can do \emph{interference management} and use the interfered aliased information to recover the signal. To convey the essential intuition, consider the following different but simpler problem: suppose $\mathsf{A}, \mathsf{B}\sim \mathcal{N}(0,\ps)$ are two independent Gaussian random variables. We are interested in recovering these two variables via one observation that is corrupted by a Gaussian noise of variance $\sigma^2$. The interference cancellation strategy corresponds to discarding $\mathsf{B}$ and observing $\mathsf{A}+Z$, where $Z\sim\mathcal{N}(0, \sigma^2)$ is the noise. This yields an estimation error of \begin{align}\ps+\frac{\ps\sigma^2}{(\ps+\sigma^2)}=\ps(1+\frac{1}{1+SNR/2}),\label{eqn:intuition1}\end{align}
where the $\ps$ and $\ps\sigma^2/(\ps+\sigma^2)$ are the estimation errors for $\mathsf{B}$ and $\mathsf{A}$, respectively, and $SNR$ is defined as $2\ps/\sigma^2$. On the other hand, an interference management scheme observes $\mathsf{A}+\mathsf{B}+Z$. Here both $\mathsf{A}$ and $\mathsf{B}$ are interference terms for each other. The total estimation error in this case is equal to
\begin{align}2\frac{\ps(\ps+\sigma^2)}{2\ps+\sigma^2}=\ps(1+\frac{1}{1+SNR}).\label{eqn:intuition2}\end{align}
In comparison to \eqref{eqn:intuition1}, the SNR gain of two is attained due to the  interference management.
}\end{discussion}

\section{Extensions}\label{sec:extensions1}

\subsection{Sparse Signals}\label{sparsesection}
Consider an extension of the results to the case when the coefficients  $A_\ell$ and $B_\ell$ in \eqref{eqdefs} are mutually independent zero-mean Gaussian r.v.s  with arbitrary variances $\ps_\ell$ for $N_1 \leq \ell \leq N_2$. 
Observe that the signal is \emph{sparse} in the frequency domain if $\ps_\ell$ is non-negligible for few values of $\ell$.  Here, we assume that the frequency support of the signal is known. 

For the case of positive equal variances $\ps_\ell=\ps>0$ for $\ell = N_1, \cdots, N_2$, we have provided two general lower bounds on the average distortion. When $\ps_\ell$ are arbitrary, we can generalize the second lower bound given in Lemma \ref{lemma:lower:2} as follows:
\begin{proposition}\label{Proposition:Sparse} 
 For any $\sigma>0$, the following lower bound on the average distortion holds for all values of $M$:
\begin{align} 
\mathsf{D}_{\min} &\geq \sum_{\ell =N_1}^{N_2} 
\frac{1}{\frac{1}{\ps_\ell}+\frac{M}{2\sigma^2} }.
\end{align}
The above inequalities are tight for some values of $M$ and $N$. The condition for equality is the possibility of choosing the sampling time instances to satisfy the following equation:
\begin{align}
\sum_{i=1}^{M} e^{j2 \pi k\frac{t_i}{T}} = 0 \qquad \emph{for} \quad k\in\{1, 2, \cdots, N-1\}\cup\{2N_1, 2N_1+1, \cdots, 2N_2\}.
\label{eqn:for:thm:unique3inlemma2sparse}
\end{align}
Uniform sampling, \emph{i.e.,}  $t_i=iT/{M}, i=1, 2, \cdots, M$ is a solution to the above equation if for each $k$ in the interval $ 2N_1 \leq k \leq 2N_2$, $M$ does not divide $k$. In particular, for the rates above the Nyquist rate, \emph{i.e.,}  $M>2N_2$, uniform sampling is optimal. 
\end{proposition}
See Appendix \ref{proof:4:5:2} for a proof.
\iffalse
Similar to the proof of Theorem~\ref{T1c}, one can find optimal sampling points for a sparse signal under certain conditions: uniform sampling with no filter on the signal bandwidth is optimal when we are above the Nyquist rate $M>2N_2$. 
\fi 

\subsection{Discrete Signals}
\label{sec:extensions}
Consider a real periodic discrete signal of the form
$S[n]=\sum_{\ell=N_1}^{N_2}[A_\ell\cos(\ell\omega_0 n)+B_\ell\sin(\ell\omega_0 n)],$
where $\omega_0=\frac{2\pi}{T}$ for some integer $T$, and $1\leq N_1\leq  N_2<T/2$ are natural numbers. Observe that $T$ is the period of the discrete signal and $A_\ell-j B_\ell$ is the $\ell$-th DFT coefficient of $S[n]$. 

Suppose that we have $M$ noisy samples at time instances $\{t_1, t_2, \cdots, t_M\}$. We would like to use these samples to reconstruct the discrete signal $S[n]$. If the reconstruction is $\hat{S}[n]=\sum_{\ell=N_1}^{N_2}[\hat{A}_\ell\cos(\ell\omega_0 n)+\hat{B}_\ell\sin(\ell\omega_0 n)]$, the distortion
$
 \sum_{n=1}^{T}|S[n]-\hat{S}[n]|^2 
$
is proportional to $
 \sum_{\ell}|A_\ell-\hat{A}_\ell|^2+|B_\ell-\hat{B}_\ell|^2 
$. Thus, the formulation for minimizing the distortion is the same as that of the continuous  signals which is given in  \ref{sec:matrix-form}. The only exception is that we additionally have the restriction that $t_i$ for $i = 1, 2, \cdots, M$ should be  integers. 
Thus, whenever the optimal sampling points in the continuous formulation turn out to be integer values, they are also the optimal points in the discrete case.  When the optimal sampling points  are not integers, simulation results with exhaustive search  show that their closest integer values are either optimal or nearly optimal.   For example, when the signal period is $T = 15$ and $(M,N_1,N)=(3, 1, 4)$. In this case $M \leq N$ and the optimal sampling points in the continuous case are $$\{ t_1, t_2, t_3\} \subset   \{0+\tau, 3.75+\tau, 7.5+\tau,  11.25+\tau \},$$
for some $\tau\in [0,T]$.  The rounds of these point for $\tau=1.1$ are optimal in the discrete case, \emph{i.e.,}   any choice of time instances from the set  $\{1, 5, 8, 12\}$ is optimal. 
On the other hand, when the signal period is $T = 15$ and $(M,N_1,N)=(4, 1, 6)$. Here again $M \leq N$ and $$\{ t_1, t_2, t_3, t_4\} \subset  \digamma =  \{0+\tau, 2.5+\tau, 5+\tau,  7.5+\tau, 10+\tau, 12.5+\tau\}$$ are the only continuous  optimal points. Rounding these points to  their closest integers yields $\mathsf{D} = 8.0069$. However, the optimal points are $\{1, 2, 8, 9\}$ resulting in $\mathsf{D}_{min} = 8.0068$. Note that the distance between the first two optimal points is 1, which cannot be achieved by choosing any arbitrary value of $\tau$, since the distance between any of the points in $\digamma$ is 2.5.

\appendix
\section*{Appendix}

\section{Estimator for Minimizing Variance of Distortion}
\label{AppendixA}
The conventional MMSE estimation problem for estimating a vector $\mathbf{X}$ from vector $\mathbf{Y}$ asks for minimizing the expected value of the distortion 
$\|\mathbf{X}-\hat{\mathbf{X}}\|^2$
where the estimator $\hat{\mathbf{X}}$ is created as a function of observation $\mathbf{Y}$. However, this would only ensure that the distortion is minimized on \emph{average}. In practice, we get one copy of $\mathbf{X}$ and $\mathbf{Y}$, and we want to ensure that the distortion that we obtain is small for that one copy, not just on average. Minimizing the variance of $\|\mathbf{X}-\hat{\mathbf{X}}\|^2$ makes sense because variance is a measure of concentration around the mean. For instance, by Chebyshev's inequality, the probability of  $\|\mathbf{X}-\hat{\mathbf{X}}\|^2$ exceeding a threshold depends both on the average of $\|\mathbf{X}-\hat{\mathbf{X}}\|^2$ and its variance. 

In this Appendix we show that for jointly Gaussian random variables, the estimator that minimizes the expected value of $\|\mathbf{X}-\hat{\mathbf{X}}\|^2$, also minimizes its variance.  More specifically, we know that 
$\hat{\mathbf{X}}=\mathbb{E}[\mathbf{X}|\mathbf{Y}]$ minimizes $\mathbb{E}(\|\mathbf{X}-\hat{\mathbf{X}}\|^2)$. We show that for \emph{ jointly Gaussian vectors}, $\hat{\mathbf{X}}=\mathbb{E}[\mathbf{X}|\mathbf{Y}]$ also minimizes the \emph{difference}
$$\mathsf{Var}(\|\mathbf{X}-\hat{\mathbf{X}}\|^2)=\mathbb{E}(\|\mathbf{X}-\hat{\mathbf{X}}\|^4)-\big(\mathbb{E}(\|\mathbf{X}-\hat{\mathbf{X}}\|^2)\big)^2.$$
\begin{theorem}\label{MMSE-var} Suppose $\mathbf{X}$ and $\mathbf{Y}$ are two correlated jointly Gaussian vectors, having covariance matrices $C_\X$, $C_\Y$ and $C_{\X\Y}$. Let $\hat{\mathbf{X}}$ be a function of $\mathbf{Y}$, and consider the cost
$$\mathsf{Var}(\|\mathbf{X}-\hat{\mathbf{X}}\|^2).$$
The estimator $\hat{\mathbf{X}}$ that minimizes the above cost constraint is $\hat{\mathbf{X}}=\mathbb{E}[\mathbf{X}|\mathbf{Y}]$, and the minimum variance  is equal to $2\tr(C_e^2)$, where
$$C_e = C_\X - C_{\X \Y} C_{\Y}^{-1} C_{\Y \X}.$$
\end{theorem}
\begin{proof} We have
\begin{align}\mathsf{Var}(\|\mathbf{X}-\hat{\mathbf{X}}\|^2)&=\mathbb{E}_\mathbf{Y}\mathsf{Var}(\|\mathbf{X}-\hat{\mathbf{X}}\|^2|\mathbf{Y})+\mathsf{Var}_\mathbf{Y}[\mathbb{E}(\|\mathbf{X}-\hat{\mathbf{X}}\|^2|\mathbf{Y})].
\end{align}
We claim that both of the terms $\mathbb{E}_\mathbf{Y}\mathsf{Var}(\|\mathbf{X}-\hat{\mathbf{X}}\|^2|\mathbf{Y})$ and $\mathsf{Var}_\mathbf{Y}[\mathbb{E}(\|\mathbf{X}-\hat{\mathbf{X}}\|^2|\mathbf{Y})]$ are minimized when 
$\hat{\mathbf{X}}=\mathbb{E}[\mathbf{X}|\mathbf{Y}]$. The second term is always non-negative and becomes zero when $\hat{\mathbf{X}}=\mathbb{E}[\mathbf{X}|\mathbf{Y}]$. This is because $\mathbb{E}(\|\mathbf{X}-\hat{\mathbf{X}}\|^2|\mathbf{Y})$ will be equal to $\mathsf{Var}(\mathbf{X}|\mathbf{Y})$ which is a constant and does not depend on the value of $\mathbf{Y}$ for jointly Gaussian random variables. Therefore, its variance is zero if  $\hat{\mathbf{X}}=\mathbb{E}[\mathbf{X}|\mathbf{Y}]$. 

To prove that the first term is minimized at $\hat{\mathbf{X}}=\mathbb{E}[\mathbf{X}|\mathbf{Y}]$, it suffices to show that this claim for all values of $\mathbf{y}$. We will be done if we have the following statement: for any jointly Gaussian vector $\mathbf{X}$, the function
$$f(\mathbf{c})=\mathsf{Var}(\|\mathbf{X}-\mathbf{c}\|^2)$$
is minimized at $\mathbf{c}=\mathbb{E}[\mathbf{X}]$. Let $\mathbf{X}\sim\mathcal{N}(\mathbf{\mu}, C_\X)$. Then observe that  $\mathbf{X}-\mathbf{c}\sim\mathcal{N}(\mathbf{\mu}-\mathbf{c}, C_\mathbf{X})$. If we replace $\mathbf{X}-\mathbf{c}$ by $U(\mathbf{X}-\mathbf{c})$ for some unitary matrix $U$, the norm $\|\mathbf{X}-\mathbf{c}\|^2$ remains invariant. Therefore, without loss of generality we can assume that $C_\mathbf{X}$ is diagonal. Then $\|\mathbf{X}-\mathbf{c}\|^2=\sum_{i}^n (X_i-c_i)^2$ is the sum of independent Gaussian random variables. Let $\sigma_i^2$ be the variance of $X_i$. Therefore
\begin{align}f(\mathbf{c})&=\sum_{i} \mathsf{Var}((X_i-c_i)^2)
\\&=\sum_{i} \mathbb{E}((X_i-c_i)^4)-\mathbb{E}((X_i-c_i)^2)^2
\\&=\sum_i (\mu_i-c_i)^4+6(\mu_i-c_i)^2\sigma_i^2+3\sigma_i^4-(\sigma_i^2+(\mu_i-c_i)^2)^2
\\&=\sum_i 4(\mu_i-c_i)^2\sigma_i^2+2\sigma_i^4,
\end{align}
which is clearly minimized when $c_i=\mu_i$, and the minimum will be equal to $2\sum_i \sigma_i^4=2\tr(C_\mathbf{X}^2)$.

Coming back to the original minimization problem, we see that the covariance matrix of $\mathbf{X}$ given any value for $\mathbf{Y}=\mathbf{y}$, does not depend on the value of $\mathbf{y}$, and is equal to $C_e$, where
$C_e = C_\X - C_{\X \Y} C_{\Y}^{-1} C_{\Y \X}.$
Therefore, the overall minimum variance is equal to 
$2\tr(C_e^2).$
\end{proof}

\section{Proofs} \label{proofs}
Before giving the proofs, we first provide a matrix representation of the problem in Section \ref{sec:matrix-form}.
\subsection{Problem Representation in a Matrix Form}\label{sec:matrix-form}

Let $\mathbf{X}$ be the vector of coefficients $A_\ell$ and $B_\ell$ ($\ell = N_1,N_1+1, \cdots, N_2$) of the original signal,      
$S(t)$,  given  in  \eqref{eqdefs}, \emph{i.e.,}  
\vspace{-0.3cm}
\begin{bbox}
\vspace{-0.4cm}
\begin{align}\mathbf{X}=[A_{N_1}, \cdots, A_{N_2}, B_{N_1}, \cdots, B_{N_2}]^\dagger.\label{eqn:def:coef:AB}\end{align}\vspace{-0.7cm}
\end{bbox}\vspace{-0.3cm}
Since the vector $\mathbf{X}$ consists of real numbers, $\dagger$ is just the transpose operation in the above equation.
After filtering $S(t)$, we get  
\begin{align}\tilde{S}(t)=\sum_{\ell=N_1}^{N_2}[\tilde{A}_\ell\cos(\ell\omega_0 t)+\tilde{B}_\ell\sin(\ell\omega_0 t)]\label{eqdefs2}
\end{align} 
in which 
\begin{align}\begin{cases}\tilde{A}_{\ell}=A_\ell H_R(\ell \omega_0)+B_\ell H_I(\ell \omega_0),\\ \tilde{B}_{\ell}=B_\ell H_R(\ell \omega_0)-A_\ell H_I(\ell \omega_0),\end{cases}\label{eqn:coeffs}\end{align} 
for $\ell = N_1,N_1+1, \cdots, N_2$, where $H_R(\omega)$ and $H_I(\omega)$ are the real and imaginary parts of $H(\omega)$, respectively. If we denote the vector of coefficients of  $\tilde{S}(t)$ as $$\tilde{\mathbf{X}}=[\tilde A_{N_1}, \cdots, \tilde A_{N_2}, \tilde B_{N_1}, \cdots,  \tilde B_{N_2}]^\dagger$$ 
and express \eqref{eqn:coeffs} in a matrix form, we get
 $$\tilde{\mathbf{X}}=L\X,$$ where $L$ is of the following form
\begin{align}L=\begin{pmatrix}{L_1} & {L_2}\\-{L_2} & {L_1}\end{pmatrix},\label{defmatrixLnew}\end{align} in which
$$L_1=\begin{pmatrix}H_R(N_1\omega_0)&0&\cdots&0\\
 0&H_R((N_1+1)\omega_0)&0&0\\
&\vdots& \\
 0& \cdots&0&H_R(N_2\omega_0)\end{pmatrix},$$
\vspace{0.01cm}
$${L_2}=\begin{pmatrix}H_I(N_1\omega_0)&0&\cdots&0\\
 0&H_I((N_1+1)\omega_0)&0&0\\
&\vdots& \\
 0& \cdots&0&H_I(N_2\omega_0)\end{pmatrix}.$$

The signal $\tilde{S}(t)$ is then sampled at time instances $t_i$ for $i = 1,2 \cdots, M$. We have
$$\tilde{S}(t_i)=\sum_{\ell=N_1}^{N_2}[\tilde{A}_\ell\cos(\ell\omega_0 t_i)+\tilde{B}_\ell\sin(\ell\omega_0 t_i)].$$ Let $\mathbf{\tilde{S}}$ be the vector of the samples
$$\mathbf{\tilde{S}}=[\tilde{S}(t_1), \tilde{S}(t_2), \cdots,  \tilde{S}(t_M)]^\dagger.$$
Therefore, in the matrix form we have $\mathbf{\tilde{S}}=Q\tilde{\X}=QL\X$, where $Q$ is an $M\times 2N$ matrix of the form
\footnotesize
$$Q=\begin{pmatrix}\cos(N_1\omega_0 t_1) & \cos((N_1+1)\omega_0 t_1) &\cdots&\cos(N_2\omega_0 t_1)
&\sin(N_1\omega_0 t_1) & \sin((N_1+1)\omega_0 t_1) &\cdots&\sin(N_2\omega_0 t_1)
\\
\cos(N_1\omega_0 t_2) & \cos((N_1+1)\omega_0 t_2) &\cdots&\cos(N_2\omega_0 t_2)
&\sin(N_1\omega_0 t_2) & \sin((N_1+1)\omega_0 t_2) &\cdots&\sin(N_2\omega_0 t_2)\\
&&&\vdots&\vdots&&&\\
\cos(N_1\omega_0 t_M) & \cos((N_1+1)\omega_0 t_M) &\cdots&\cos(N_2\omega_0 t_M)
&\sin(N_1\omega_0 t_M) & \sin((N_1+1)\omega_0 t_M) &\cdots&\sin(N_2\omega_0 t_M)
\end{pmatrix}.$$\normalsize
Finally, the vector of observations denoted by $\mathbf{Y}=\mathbf{\tilde{S}}+\mathbf{Z}$ can be written as
\vspace{-0.3cm}
\begin{bbox}
\vspace{-0.2cm}$$\mathbf{Y}=QL\mathbf{X}+\mathbf{Z},$$\vspace{-0.7cm}\end{bbox}\vspace{-0.3cm}
where $\mathbf{Z}$ is the noise vector. 

Estimation of $S(t)$ is equivalent to the estimation of its Fourier coefficients. If we denote the coefficients of the reconstruction signal by 
$$\hat{\mathbf{X}}=[\hat{A}_{N_1}, \cdots, \hat{A}_{N_2}, \hat{B}_{N_1}, \cdots, \hat{B}_{N_2}]^\dagger,$$ using the Parseval's theorem, the sampling distortion is equal to
\begin{align}\frac{1}{T}\int_{t=0}^T|\hat{S}(t)-S(t)|^2dt &= \frac12 \sum_{\ell=N_1}^{N_2}\left(|\hat{B}_\ell-B_\ell|^2+|\hat{A}_\ell-A_\ell|^2\right)\nonumber
\\&=\frac 12\|\mathbf{X}-\hat{\mathbf{X}}\|^2\label{eqn:Parseval},\end{align}
which is a random variable. Our goal is to minimize its average and variance. 

\textbf{Computing average distortion:}
We use $\mathbf{Y}$ to estimate the signal with minimal distortion. In other words, from  
\eqref{eqn:Parseval}, we would like to minimize 
\begin{align}\mathsf{D}=\frac{1}{T}\int_{t=0}^T\mathbb{E}\{|\hat{S}(t)-S(t)|^2\}dt=\frac 12 \mathbb{E}\|\mathbf{X}-\hat{\mathbf{X}}\|^2.\label{eqn:Dvalue1}\end{align}
 Since all the random variables are Gaussian, the linear MMSE is optimal and thus we would like to use $QL\mathbf{X}+\mathbf{Z}$ to find $\hat{\mathbf{X}}$ such that $E\|\mathbf{X}-\hat{\mathbf{X}}\|^2$ is minimized.

The mean square error is given by
\begin{align}
\mathbb{E}\| \mathbf{X}-\hat{\mathbf{X}} \|^2 & =\mathbb{E}_{\mathbf{Y}}\left \{ \mathsf{Var}[\mathbf{X} | \mathbf{Y}]\right\}=\tr(C_e), \label{eqn:mmseCe1}
\end{align}
where the error covariance matrix $C_e$ is of the form 
\begin{align} C_e &=  C_{\X} - C_{\X\Y} C_{\Y}^{-1} C_{\Y\X}
\\& =  C_{\X}- C_{\X}  L^\dagger Q^\dagger(Q L C_{\X} L^\dagger Q^\dagger+C_{\Z})^{-1} Q L C_{\X}   \nonumber
\\ &=\ps I- \ps^2L^\dagger Q^\dagger(\ps QLL^\dagger Q^\dagger + \sigma^2 I)^{-1}QL \label{Ce1}.\end{align}
In the above formulas, we have used the fact that $C_{\X}= \ps I_{2N \times 2N} $ and $C_{\Z} = \sigma^2 I_{M \times M}$. Therefore, from \eqref{eqn:Dvalue1} \vspace{-0.3cm}\begin{bbox}\vspace{-0.3cm}$$\mathsf D=\tr(C_e)/2$$\vspace{-0.7cm}\end{bbox}\vspace{-0.3cm}  and if we use  \eqref{Ce1}, it follows that:
\begin{align}2\mathsf{D}=\tr(C_e) &=\tr\left(\ps I- \ps^2 L^\dagger Q^\dagger (\ps QLL^\dagger Q^\dagger + \sigma^2 I)^{-1} QL\right)\label{eqn:eqstart1}
\\ &=2N\ps -\ps \cdot\tr\left(\ps L^\dagger Q^\dagger(\ps QLL^\dagger Q^\dagger + \sigma^2 I)^{-1} QL\right)\nonumber
\\ &=2N\ps -\ps \cdot\tr\left( (\ps QLL^\dagger Q^\dagger + \sigma^2 I)^{-1} (\ps QLL^\dagger Q^\dagger)\right)\label{Eq2}
\\ &=2N\ps -\ps \cdot\tr\left( (\ps QLL^\dagger Q^\dagger + \sigma^2 I)^{-1} (\ps QLL^\dagger Q^\dagger+\sigma^2I-\sigma^2I)\right)\nonumber
\\&=2N\ps -\ps \cdot\tr\left (I-\sigma^2 (\ps QLL^\dagger Q^\dagger + \sigma^2 I)^{-1} \right)\label{Sigmagoestozero}
\\&=(2N-M)\ps +\ps\sigma^2\tr\left((\ps QLL^\dagger Q^\dagger + \sigma^2 I)^{-1}\right) \label{eqn:Distortion1},
\end{align}
where  \eqref{Eq2} results from the fact that  the trace operator is invariant under cyclic permutations.

There is an alternative form of  the Linear MMSE estimator (\cite[p.~90]{Luenberger}) in which  the error covariance matrix $C_e$ is of the form 
\begin{align}  C_e & = (L^\dagger Q^\dagger C_{\Z}^{-1}  QL+C_{\X}^{-1}  )^{-1} \label{Ce2half}
\\&= \ps\sigma^2 (\ps L^\dagger Q^\dagger QL + \sigma^2 I)^{-1} \label{Ce2}.\end{align}
Equivalently, $\mathsf{D}$ can be found using  \eqref{Ce2} as
\begin{align}2\mathsf{D}=\tr(C_e) &=\ps \sigma^2\tr (\ps L^\dagger Q^\dagger QL + \sigma^2 I)^{-1} \label{eqn:Distortion2}.
\end{align}
\vspace{-0.7cm}
\begin{ybox}
\textbf{When $L=I$:} We can express the two distortion formulas given in \eqref{eqn:Distortion1} and \eqref{eqn:Distortion2} in terms of the following two matrices:
$$\Delta_1=\ps Q Q^\dagger  + \sigma^2 I,$$
$$\Delta_2=\ps Q^\dagger Q + \sigma^2 I$$
as follows:
\begin{align}2\mathsf{D}&=(2N-M)\ps +\ps\sigma^2\tr\left(\Delta_1^{-1}\right)\label{eqn:Distortion2mmmm1}
\\&=\ps \sigma^2\tr \Delta_2^{-1} \label{eqn:Distortion2mmmm2}.
\end{align}
A direct calculation shows that the entries of matrix $\Delta_1$ (an $M\times M$ matrix) are:
\begin{align}\Delta_{1(i,k)}&=\ps\left(\sum_{\ell=N_1}^{N_2} \cos(\ell\omega_0(t_i-t_k))\right)+\sigma^2 \mathbf{1}[i=k], \label{eqn:PiEntries}\end{align}
where $\mathbf{1}[i=k]$ is one if $i=k$ and zero otherwise.
Moreover, a direct calculation shows that the diagonal entries of matrix $\Delta_2$ (a $2N\times 2N$ matrix) are given by:
\begin{align}
\Delta_{2(k,k)} =\begin{cases}\sigma^2+\ps \sum_{i=1}^{M}\cos^2(\ell\omega_0
t_i);~~\text{where }\ell = k+N_1-1 \text{ for } k  = 1, \cdots , N,
\\
\sigma^2+\ps \sum_{i =1}^{M}\sin^2(\ell\omega_0 t_i);
  ~~\text{where }\ell = k+N_1-1-N\text{ for } k= N+1, \cdots , 2N.\end{cases} \label{eqn:GammaEntries}\end{align}
\vspace{-0.3cm}
\end{ybox}
\vspace{-0.3cm}

\normalsize
\textbf{Computing variance of distortion:}

The variance of distortion, using the Parseval's theorem given in  \eqref{eqn:Parseval}, is of the form
\begin{align}
\mathsf{V}&=\mathsf{Var}\left\{\frac{1}{T}\int_{t=0}^T|\hat{S}(t)-S(t)|^2dt\right\}
=\frac 12 \mathsf{Var}\{\|\mathbf{X}-\hat{\mathbf{X}}\|^2\}.
\end{align}
Using Theorem \ref{MMSE-var} from Appendix \ref{AppendixA}, we have 
\begin{align}
\mathsf{Var}(\|\X-\hat{\X}\|^2) =2\tr(C_e^2).  \end{align}
Hence, the variance of distortion will be
\vspace{-0.3cm}
\begin{bbox}
\vspace{-0.2cm}
$$\mathsf{V}=\tr(C_e^2).$$\vspace{-0.7cm}
\end{bbox}\vspace{-0.3cm}
Using similar steps as the ones used in computing average distortion, from \eqref{Ce1},  $\mathsf{V}$ is found to be 
\begin{align}
\mathsf{V} = \frac 12 \mathsf{Var}(\|\X-\hat{\X}\|^2) &= \tr(C_e^2)   \nonumber
\\ & =\tr[\left(\ps I- \ps^2 L^\dagger Q^\dagger (\ps QLL^\dagger Q^\dagger + \sigma^2 I)^{-1}  QL\right)^2]
\\&=\ps^2 [2N-M  + \sigma^4 \tr((\ps QLL^\dagger Q^\dagger + \sigma^2 I)^{-2})], \label{Var1}
\end{align}
where the above equality holds because by defining $\Pi_{M\times 2N}=QL$, we can use the identity
$$\tr[\left(\ps I- \ps^2 \Pi^\dagger (\ps \Pi\Pi^\dagger + \sigma^2 I)^{-1}  \Pi\right)^2]=
\ps^2 [2N-M  + \sigma^4 \tr((\ps \Pi\Pi^\dagger + \sigma^2 I)^{-2})].$$
Alternatively, using the second form of MMSE given in \eqref{Ce2}, we get 
\begin{align}
\mathsf{V} = \tr(C_e^2) = \ps^2 \sigma^4 \tr(\ps L^\dagger Q^\dagger QL + \sigma^2 I)^{-2}
 \label {Var2}.\end{align}
\vspace{-0.3cm}
\begin{ybox}
\textbf{When $L=I$:} We can express the variance formulas in terms of matrices
$$\Delta_1=\ps Q Q^\dagger  + \sigma^2 I,\qquad \Delta_2=\ps Q^\dagger Q + \sigma^2 I$$
as follows:
\begin{align}\mathsf{V} &=\ps^2 [2N-M  + \sigma^4 \tr(\Delta_1^{-2})] \label{eqn:Var2mmmm1}
\\&=\ps^2 \sigma^4 \tr(\Delta_2^{-2}).\label{eqn:Var2mmmm2}
\end{align}\vspace{-0.5cm}
\end{ybox}
\vspace{-0.3cm}
So far, two closed formulas for $\mathsf{D}$ and $\mathsf{V}$ have been derived in equations \eqref{eqn:Distortion2mmmm1}, \eqref{eqn:Distortion2mmmm2}, \eqref{eqn:Var2mmmm1} and \eqref{eqn:Var2mmmm2}. In the following subsections, the proof of the main results will be developed using the notation and formulas given in  this subsection.
\subsection{Proof of Theorem \ref{thm:pre-sampling} (Section \ref{sec:main-results})}
We would like to show that setting $L=I$ is optimal. The coordinates of matrix $L$ are determined by the LTI filter as given in equation \eqref{defmatrixLnew}. Observe that
\begin{align}LL^\dagger= \begin{pmatrix}{L_1} & {L_2}\\-{L_2} & {L_1}\end{pmatrix}\begin{pmatrix}{L_1} & -{L_2}\\{L_2} & {L_1}\end{pmatrix}=\begin{pmatrix}{L_1}^2+ {L_2}^2&0\\0 & {L_1}^2+ {L_2}^2\end{pmatrix}.\label{eqnLLdaggerlessI}\end{align}
The matrix \small
$$L_1^2+L_2^2=\begin{pmatrix}H_R^2(N_1\omega_0)+H^2_I(N_1\omega_0)&0&\cdots&0\\
 0&H^2_R((N_1+1)+H^2_I((N_1+1)\omega_0)&0&0\\
&\vdots& \\
 0& \cdots&0&H_R^2(N_2\omega_0)+H_I^2(N_2\omega_0)\end{pmatrix}$$\normalsize
has diagonal entries that are less than or equal to one, by the passive filter condition. Therefore, from \eqref{eqnLLdaggerlessI} and the assumption that the filter is passive, we obtain the constraint $LL^\dagger \leq I$,  where by $A\leq B$ we mean that $B-A$ is positive semi-definite.

Hence, $Q(LL^\dagger-I) Q^\dagger \leq 0$, implying that
\begin{align}\ps QLL^\dagger Q^\dagger + \sigma^2 I\leq \ps QQ^\dagger+ \sigma^2 I.\label{eqn:ordernewsq}\end{align}
We now use the following inequality:

\begin{theorem}[Klein's inequality]\cite{Bhatia} For any two symmetric positive definite matrices $A$ and $B$ and any differentiable convex function on $(0,\infty)$, we have 
$$\tr(f(A)-f(B)]\geq \tr((A-B)f'(B)).$$
\end{theorem}
Let $f(x)=x^{-1}$ and $A=\ps QQ^\dagger+ \sigma^2 I$ and $B=\ps QLL^\dagger Q^\dagger + \sigma^2 I$. Then from \eqref{eqn:ordernewsq}, we have $B\geq A$. Now, observe that $(A-B)f'(B)=(B-A)B^{-2}$ is the product of two positive semi-definite matrices. Hence, its trace is non-negative.\footnote{Observe that if $A$ and $B$ are positive semi-definite, then $\tr(AB)=\tr(AB^{1/2}B^{1/2})=\tr(B^{1/2}AB^{1/2})\geq 0.$} Thus,
$\tr[f(A)-f(B)]\geq 0.$ We get that
\begin{align}\tr (\ps QLL^\dagger Q^\dagger + \sigma^2 I)^{-1} \geq \tr (\ps QQ^\dagger + \sigma^2 I)^{-1} .\end{align}
Not using any filter on the signal bandwidth is equivalent with $LL^\dagger=I$. Hence, pre-filtering is not helpful in reducing the MMSE. The pre-filtering is also not helpful in reducing variance of the distortion. It is shown that 
$$\mathsf{Var}(\|\X-\hat{\X}\|^2)=2\ps^2(2N-M)  + 2\ps^2\sigma^4 \tr\left((\ps QLL^\dagger Q^\dagger + \sigma^2 I)^{-2}\right),$$
and one can employ Klein's inequality for $f(x)=x^{-2}$ in the same manner.

\subsection{Proof of Lemma \ref{lemma:lower:1} (Section \ref{sec:lemma:1:lower:ab})} \label{Lower1proof}
Since the noise variance $\sigma^2 >0$, the matrix $\Delta_1=\ps QQ^{\dagger}+ \sigma^2 I$ is positive definite and real. In addition, from \eqref{eqn:PiEntries}, all the diagonal entries of $\Delta_1$ are  $N \ps +\sigma^2$.
Applying Theorem \ref{thm:lemma:diag} for the matrix $\Delta_1$ and the convex function $f(x)=x^{-1}$ on $x>0$, we have  
\begin{align}\tr\left(\Delta_1^{-1}\right)\geq \sum_{j=1}^M \Delta_{1(j,j)}^{-1}= \frac{M}{N\ps +\sigma^2}.\label{Eq6prime}\end{align}
Substituting \eqref {Eq6prime} into \eqref{eqn:Distortion2mmmm1}, results in a lower bound on distortion as follows 
\begin{align*}2\mathsf D&=(2N-M)\ps +\ps \sigma^2\tr\left(\Delta_1^{-1}\right)
\\&\geq (2N-M)\ps +\ps \sigma^2 \frac{M}{N\ps +\sigma^2}
\\&= (2N-M)\ps +\ps \frac{M}{1+SNR}.
\end{align*}
Note that we have defined $SNR = N\ps/\sigma^2$. 
Because sampling time instances $t_i$ are arbitrary,  for any value of $M$, we obtain
\begin{align}\frac{\mathsf{D}_{\min}}{\ps}&\geq  \frac 12 \left(2N-M+\frac{M}{1+SNR}\right).\label{LB1}
\end{align}
Since $f(x)=x^{-1}$ is a strictly convex function for $x>0$, equality in the above equation holds if and only if $\Delta_1=\ps Q Q^{\dagger}+ \sigma^2 I$ is a diagonal matrix. Therefore, using \eqref{eqn:PiEntries}, we have the following equation for $1\leq k< i \leq M$ (remember that $\omega_0=2\pi/T$):
\begin{align*}
\Delta_1(i,k) &=  \sum_{\ell=N_1}^{N_2}\cos(\ell\omega_0(t_i-t_k))=\mathsf{Real} \{ \sum_{\ell=N_1}^{N_2} e^{j2\pi\ell\frac{t_i-t_k}{T} }\} 
\\&= \mathsf{Real} \{  e^{j\pi \frac{t_i-t_k}{T} (N_1+N_2)}. ~\frac{\sin(\pi  \frac{t_i-t_k}{T} N)}{\sin(\pi  \frac{t_i-t_k}{T} )} \}
\\&=\cos(\pi \frac{t_i-t_k}{T} (N_1+N_2)) .~ \frac{\sin(\pi  \frac{t_i-t_k}{T} N)}{\sin(\pi  \frac{t_i-t_k}{T} )}=0.
\end{align*}
Consequently, the  sampling time instances $t_i$ should satisfy the following equations: 
\begin{align}
 \begin{cases}
\sin(\pi N \frac{(t_i - t_k)}{T}) = 0, \\ 
\qquad\text{or} \\
\cos(\pi (N_1+N_2) \frac{(t_i - t_k)}{T}) = 0. 
\end{cases} \label{eqn:sincos1}\end{align}
Thus, given any $i$ and $j$, we should either have
\begin{align*}|t_i - t_j|&=T\frac {m_1}{N},& \qquad~~\quad\text{for some integer } m_1,\end{align*}or\begin{align*}
|t_i - t_j|&=T\frac{2m_2+1}{2(N_1+N_2)},& \text{for some integer }  m_2.
\end{align*}

\emph{Derivation of the lower bound for variance:}\\
 Similarly, the lower bound on the  variance is developed using \eqref{eqn:Var2mmmm1} and Theorem \ref{thm:lemma:diag} for the convex function $f(x)=x^{-2}$ on $x>0$. From Theorem \ref{thm:lemma:diag}, we have
\begin{align}\tr\left(\Delta_1^{-2}\right) \geq
\sum_{j=1}^M \Delta_{1(j,j)}^{-2} = \frac{M}{(N\ps +\sigma^2)^2}.\end{align}
And thus using  \eqref{eqn:Var2mmmm1} , we can conclude that 
$$ \mathsf{V}_{\min} \geq \ps^2 \left [2N-M  + \sigma^4 \frac{M}{(N \ps +\sigma^2)^2}\right] =\ps^2 \left[2N-M  + \frac{M}{(1 +SNR)^2}\right] .$$
Here again,  equality holds if and only if $\Delta_1=\ps QQ^{\dagger}+ \sigma^2 I$ is a diagonal matrix.

\subsection{Proof of Theorem \ref{T1a} (Section \ref{sec:lemma:1:lower:ab})}\label{proof:4:2:1} 
Here we are in the case of $M\leq N$. In Lemma \ref{lemma:lower:1}, the desired lower bounds on $\mathsf{D}_{\min}$ and $\mathsf{V}_{\min}$ were developed. Moreover, we showed that these lower bounds are achievable if for any time instances $t_i$ or $t_j$ 
\begin{align}|t_i - t_j|&=T\frac {m_1}{N},& \text{for some integer } m_1,\label{optimaltimes1}\\ or~~~~
|t_i - t_j|&=T\frac{2m_2+1}{2(N_1+N_2)},& \text{for some integer }  m_2.\label{optimaltimes2}
\end{align} 
From \eqref{optimaltimes1}, we conclude that any arbitrary choice of time instances from the set
$$\{\tau, \tau+\frac{1}{N}T, \tau+\frac{2}{N}T, \tau+\frac{3}{N}T, \cdots, \tau+\frac{N-1}{N}T\}$$
 is optimal.

To find the interpolation formula, consider the MMSE estimator of $\X$ from $\Y=Q\mathbf{X}+\mathbf{Z}$. The estimator $\hat{\X}$ that minimizes $E\|\mathbf{X}-\hat{\mathbf{X}}\|^2$ can be written as $W\Y$, where $W=C_{\X\Y}C_{\Y}^{-1}=\ps Q^\dagger \Delta_1^{-1}$.
Since  with the optimal choice of $t_i$, $\Delta_1=(N\ps+\sigma^2)I$, the Fourier coefficients of the reconstructed signal are as follows:
\begin{align}\hat{\X}=\frac{\ps}{N\ps+\sigma^2}Q^\dagger\Y\label{eqn:hatXrec1}.\end{align}
The above formula for $\hat{\X}$ results in  the desired reconstruction formula in \eqref{eqn:Interpolation1}.

\subsection{Proof of Proposition  \ref{thm:unique} (Section \ref{sec:lemma:1:lower:ab}) }\label{proof:4:2:3}
Remember the lower bounds given in Lemma \ref{lemma:lower:1}. We showed that equality in this lemma holds if and only if 
\begin{align}|t_i - t_j|&=T\frac mN,& \text{for some integer } m, \label{eqn:case1eq}\\
|t_i - t_j|&=T\frac{2m+1}{2(N_1+N_2)},& \text{for some integer }  m.\label{eqn:case2eq}
\end{align}
It remains to show that when $M\leq N$ and $N$ does not divide $N_1+N_2$, we have the following statement: one can find $M-1$ sampling times such that their pairwise differences satisfy  \eqref{eqn:case1eq}. Without loss of generality assume that $t_1=0$. For any $t_i, t_j$, the  differences $t_i-t_1$, $t_j-t_1$ cannot simultaneously satisfy  \eqref{eqn:case2eq}, since then $t_i-t_j$ will be of the form $T(2m_i-2m_j)/(2(N_1+N_2))$ and this pairwise distance does not satisfy either of \eqref{eqn:case1eq} or \eqref{eqn:case2eq}. Therefore, at least $M-1$ points should satisfy $t_i-t_1$ with  \eqref{eqn:case1eq}. This completes the proof.

\subsection{Proof of Lemma \ref{lemma:lower:2} (Section \ref{sec:lemma2:section:abs})} \label{Lower2proof}

We use the distortion formula given in  \eqref{eqn:Distortion2mmmm2}  and minimize the distortion subject to the sampling locations. We have 
\begin{align} 
\mathsf{D}_{\min}&=
{\min_{\{t_i, i = 1, \cdots, M\}}}\frac12 \ps\sigma^2\tr(\Delta_2^{-1})\nonumber
\\&\geq \frac12\ps \sigma^2  {\min_{\{t_i, i = 1, \cdots, M\}}} \sum_{k =1}^{2N} \Delta_{2(k,k)}^{-1}
\label{Eq. 1}
 \\&= \frac12\ps \sigma^2   {\min_{\{t_i, i = 1, \cdots, M\}}}\sum_{\ell =N_1}^{N_2} \frac{1}{\eta_\ell}+\frac{1}{\gamma_\ell},\label{Eq. 6}
 \\&\geq \frac12\ps \sigma^2   {\min_{\{t_i, i = 1, \cdots, M\}}}\sum_{\ell =N_1}^{N_2} \frac{2}{(\eta_\ell+\gamma_\ell)/2}\label{Eq. 6mmm}
\\&=\frac12 \ps\sigma^2\sum_{\ell =N_1}^{N_2} 
\frac{2}{\sigma^2+\frac {M\ps }{2} }=\frac{N\sigma^2\ps }{\sigma^2+\frac {M\ps }{2}}=\frac{N\ps }{1+\frac {M}{2N} SNR },
\end{align}
where \eqref{Eq. 1} follows from Theorem \ref{thm:lemma:diag} for the convex function $f(x)=x^{-1}$ for $x>0$; \eqref{Eq. 6} is written  using \eqref{eqn:GammaEntries} where  $\eta_\ell$ and $\gamma_\ell$ are defined as
\begin{align}\eta_\ell=\sigma^2+\ps\sum_{i =1}^{M}\cos^2(\ell\omega_0 t_i),\qquad \gamma_\ell=\sigma^2+\ps\sum_{i =1}^{M}\sin^2(\ell\omega_0 t_i).\label{eqn:bet-1mmm}\end{align}
Furthermore,  \eqref{Eq. 6mmm} follows from the inequality $a^{-1}+b^{-1}\geq 4(a+b)^{-1}$  for non-negative $a$ and $b$. 

The proof for the variance is similar. Using  \eqref{Var2}, and Theorem \ref{thm:lemma:diag} for the convex function $f(x)=x^{-2}$ for $x>0$,  we have
\begin{align} 
\mathsf{V}_{\min}&= {\min_{\{t_i, i = 1, \cdots, M\}}}[\ps^2 \sigma^4 \tr(\Delta_2^{-2})]\nonumber
\\&\geq \ps^2 \sigma^4  {\min_{\{t_i, i = 1, \cdots, M\}}} \sum_{k =1}^{2N} \Delta_{2(k,k)}^{-2}\nonumber
 \\&= \ps^2 \sigma^4  {\min_{\{t_i, i = 1, \cdots, M\}}}\sum_{\ell =N_1}^{N_2} \frac{1}{\eta^2_\ell}+\frac{1}{\gamma^2_\ell}
\\&\geq \ps^2 \sigma^4  {\min_{H(.),\{t_i, i = 1, \cdots, M\}}}\sum_{\ell =N_1}^{N_2} \frac{8}{(\eta_\ell+\gamma_\ell)^2}\label{vareqnew22}
\\&\nonumber = \ps^2 \sigma^4\sum_{\ell =N_1}^{N_2} \frac{8}{(2\sigma^2+M\ps)^2}
=
N\ps^2 \sigma^4 \frac{8}{(2\sigma^2+M\ps)^2}\nonumber
=N\ps^2  \frac{2}{(1+\frac{M\ps}{2\sigma^2})^2}\nonumber
=\frac{2N\ps^2 }{(1+\frac {M}{2N} SNR)^2 }\nonumber
\nonumber,\end{align}
where  $\eta_\ell$ and $\gamma_\ell$ are given in  \eqref{eqn:bet-1mmm}, and  \eqref{vareqnew22} follows from $a^{-2}+b^{-2} \geq 8(a+b)^{-2},$ for nonzero values of $a$ and $b$.

\emph{Necessary and Sufficient conditions for tightness of the lower bounds:} For Theorem \ref{thm:lemma:diag} to be tight, 
we need $\Delta_2=\ps Q^{\dagger} Q +\sigma^2 I$ to be a  diagonal matrix. We further need $\eta_\ell=\gamma_\ell$ to have all the inequalities as equalities.  Therefore, the equality holds if and only if matrix $\Delta_2$ is  diagonal, and all the diagonal entries are equal. \\
The off-diagonal  and diagonal entries of $\Delta_2$ are given by 
$$\Delta_{2(i,k)}= \begin{cases}
  \ps \cdot \sum_{\ell=1}^{M}\cos((N_1+i-1)\omega_0t_\ell) \cos((N_1+k-1)\omega_0t_\ell)~~~;~1 \leq i<k \leq N\\
  \ps \cdot \sum_{\ell=1}^{M}\cos((N_1+i-1)\omega_0t_\ell) \sin((N_1+k-N)\omega_0t_\ell)~~;~1 \leq i\leq N\leq k \leq2 N\\
   \ps \cdot \sum_{\ell=1}^{M}\sin((N_1+i-N)\omega_0t_\ell) \sin((N_1+k-N)\omega_0t_\ell)~;~ N\leq i \leq2 N,~ N\leq k \leq2 N\\
  \end{cases}, $$
\begin{align}\Delta_2{(i,i)}= \begin{cases}
\frac{M\ps }{2}+\sigma^2 + \frac{\ps}{2}\cdot \sum_{\ell=1}^{M}\cos(2(N_1+i-1)\omega_0t_\ell)~;~ 1\leq i\leq N\\
\frac{M\ps }{2}+\sigma^2 - \frac{\ps}{2}\cdot \sum_{\ell=1}^{M}\cos(2(N_1+i-1)\omega_0t_\ell)~;~ N\leq i \leq2N \end{cases},\label{eqn:newamnewamam}\end{align} respectively.
If we put the off-diagonal entries zero and write the above equations in a simpler form, we get
\begin{align}
\begin{cases}
\sum_{\ell=1}^{M}\cos \left((N_1+k_1)\omega_0t_\ell \right)\cos \left((N_1+k_2)\omega_0t_\ell \right)=0~ ~;~~0 \leq k_1<k_2 \leq N-1\\
 \sum_{\ell=1}^{M}\cos \left((N_1+k_1)\omega_0t_\ell\right)\sin \left((N_1+k_2)\omega_0t_\ell \right)=0~~;~~0 \leq k_1, k_2\leq N-1\\
\sum_{\ell=1}^{M}\sin \left((N_1+k_1)\omega_0t_\ell \right)\sin \left((N_1+k_2)\omega_0t_\ell \right)=0~~;~~ 0\leq k_1 < k_2 \leq N-1\\
  \end{cases}. 
\end{align}
  Substituting $\omega_0= 2\pi /T$ in the above equations,  we obtain
$$\begin{cases}
\sum_{\ell=1}^{M}e^{j2\pi (k_1-k_2)\frac{t_\ell}{T}}=0~ ~;~~0 \leq k_1<k_2 \leq N-1\\ \sum_{\ell=1}^{M}e^{j2\pi (2N_1+k_1+k_2)\frac{t_\ell}{T}}=0~~;~~0 \leq k_1\leq k_2\leq N-1
  \end{cases},$$ or in another form, we have
\begin{align}
\begin{cases}
\sum_{i=1}^{M} e^{j2 \pi k\frac{t_i}{T}}=0  & \text{for} ~~0 < k \leq N-1\\
\sum_{i=1}^{M} e^{j2 \pi k\frac{t_i}{T}}=0 & \text{for} ~~2N_1 \leq k\leq 2N_2
\end{cases}.\label{eqn:solutionequality}\end{align}
These equations imply that the off-diagonal entries are all zero, and specifically,
$$\sum_{\ell=1}^{M}\cos(2(N_1+i-1)\omega_0t_\ell)=0, \qquad  1\leq i \leq2N.$$
Thus, from \eqref{eqn:newamnewamam}, the diagonal entries $\Delta_2(i,i)$ are all equal to $M\ps /2+\sigma^2$.

\subsection{Proof of Theorem \ref{T1c} (Section \ref{sec:lemma2:section:abs})}\label{proof:4:4:1}
Consider the proof of Lemma \ref{lemma:lower:2} in which the alternative   lower bound  on the average and variance of  distortion, and the necessary and sufficient conditions for their tightness were derived. With uniform sampling, \emph{i.e.,}  $t_i=\frac{iT}{M}, i=1, 2, \cdots , M$, when $2N_1+\alpha+\beta$ does not divide $M$ for integers $0\leq\alpha,\beta\leq N-1$,  matrix $\Delta_2$ will be diagonal with diagonal entries $(\sigma^2 + \frac {M\ps }{2})$. Thus this lower bounds will be tight. \\
To find the interpolation formula, consider the MMSE estimator of $\X$ from $\Y=Q\mathbf{X}+\mathbf{Z}$. The estimator $\hat{\X}$ that minimizes $E\|\mathbf{X}-\hat{\mathbf{X}}\|^2$ can be written as $W\Y$, where $W=C_{\X\Y}C_{\Y}^{-1}=\ps \Delta_2^{-1}Q^\dagger$. For the optimal sampling locations $t_i$, $\Delta_2  = (Mp/2+\sigma^2)I$, and hence
\begin{align}
 \hat{\X} = W \Y= \frac{\ps}{\frac{M\ps}{2}+\sigma^2}  Q^\dagger  \Y.
\end{align}
The estimated coefficients vector, $\hat{\X}$, results in the reconstruction formula   given in \eqref{eqn:Interpolation2}.

\subsection{Proof of Theorem \ref{T:NM2N} (Section \ref{sec:4:3})}\label{proof:4:3}
Here we are in the case of $N< M\leq 2N$. To prove the upper bound, we choose not to utilize a pre-sampling filter and use the following $M$ sampling points $$\{0, \frac{1}{N}T, \frac{2}{N}T, \frac{3}{N}T, \cdots, \frac{N-1}{N}T,~~\frac{1}{2N}T, \frac{3}{2N}T,\cdots, \frac{2M-2N-1}{2N}T\}.$$
Using  \eqref{eqn:PiEntries} , one can verify that 
$$\Delta_1=\left[\begin{matrix}{(N\ps +\sigma^2 )~ \textbf{I}_{N\times N}} && { \ps\textbf{~G}_{N\times (M-N)}}
\\
\\
{ \ps\textbf{~G}^{\dagger}_{(M-N)\times N}} && {(N\ps +\sigma^2 )~\textbf{I}_{(M-N)\times (M-N)})}
\end{matrix}\right]_ { M\times M },$$
where $\textbf{G}$  is  an $N\times (M-N)$ matrix whose $(M-N)$  columns  are the first $(M-N)$ columns of the matrix $\Phi_{N\times N}$ defined as follows: $\Phi$ is an $N\times N$ circulant matrix with the first row $\textbf {c} = [c_0,~ c_1, ~... ,~c_{N-1}]$, where
\begin{align}c_k = \frac 12\sum_{\ell=N_1}^{N_2}\big(\omega ^{-(2k+1)\ell}+\omega ^{(2k+1)\ell}\big)\label{Eq5}
\end{align}
 in which $\omega =\exp(j{\pi/N} )$. For instance, when $N_1=1$, $\Phi$ is an $N\times N$ matrix with all $-1$ entries.

Using Theorem \ref{L3} from Appendix \ref{sec:B81}, the eigenvalues of $\Delta_1-(N\ps +\sigma^2)\mathbf{I}$ are equal to $p s_i$ and $-ps_i$ (where ${s_i}$ are the singular values of $\textbf{G}_{N \times (M-N)}$), in addition to $M-2(M-N)=2N-M$ zero eigenvalues. One can find the eigenvalues of $\Delta_1$ by adding $N\ps +\sigma^2$ to the eigenvalues of $\Delta_1-(N\ps +\sigma^2)\mathbf{I}$, and from that we have
\begin{align}\tr(\Delta_1^{-1})&=\frac{2N-M}{N\ps +\sigma^2}+\sum_{i=1}^{M-N}\big(\frac{1}{N\ps +\sigma^2-ps_i}+\frac{1}{N\ps +\sigma^2+ps_i}\big)\nonumber
\\&=\frac{2N-M}{N\ps +\sigma^2}+\sum_{i=1}^{M-N}\frac{2(N\ps +\sigma^2)}{(N\ps +\sigma^2)^2-p^2 s_i^2}.\label{Eq7}\end{align}
To compute an upper bound on $\tr(\Delta_1^{-1})$, we need to find an upper bound on $s_i$. To do so, first we assume that $M=2N$, and find the singular values of the circulant matrix $\Phi_{N\times N}$ (named $\delta_i$ for $i= 1, \cdots, N$) and then use the upper bound given in Theorem \ref{T4} on the singular values of the matrix $\textbf{G}_{N \times (M-N)}$.

Singular values of $\Phi$ are the eigenvalues of $(\Phi\Phi^{\dagger})^{1/2}$. Theorem \ref{T2}  from Appendix \ref{sec:B81} states that any two circulant matrices commute and the eigenvalues of their product is the pairwise product of their eigenvalues. Since both $\Phi$ and  $\Phi^\dagger$ are circulant matrices, we conclude that $\Phi\Phi^{\dagger}$ is also a circulant matrix with eigenvalues $|\lambda_i|^2, i = 1, ..., N $, where $\lambda_i$ are eigenvalues of $\Phi$. Thus, $\delta_i= |\lambda_i|$ for $i = 1, \cdots , N.$
Using Theorem \ref{T3}  from Appendix \ref{sec:B81}, $ \lambda_i$ are given by  
\begin{align} \lambda_i =\sum_{k=0}^{N-1} c_k \omega^{-2ki}, \qquad i=1, 2, \cdots, N \label{lambda}\end{align}
where $\omega =\exp(j{\pi/N} )$.
 Substituting $c_k$ from \eqref{Eq5} into  \eqref{lambda}, we obtain 
\begin{align*} \lambda_i &=\frac 12\sum_{k=0}^{N-1}  ~\big[\sum_{\ell=N_1}^{N_2}\big(\omega ^{-(2k+1)\ell}+\omega ^{(2k+1)\ell}\big) ]~ \omega^{-2ki}
\\&=\frac 12\sum_{\ell=N_1}^{N_2} \omega ^{-\ell} \sum_{k=0}^{N-1} \omega ^{-2k(\ell+i)}+\frac 12\sum_{\ell=N_1}^{N_2}\omega ^{\ell} \sum_{k=0}^{N-1} \omega ^{2k(\ell-i)}.
\end{align*}
We then have
\begin{align}\sum_{k=0}^{N-1} \omega^{-2k(\ell+i)}= \begin{cases} N ~~~ ;  ~~~~N~|\ell+i\\0 ~~~~ ;  ~~~~N\nmid \ell+i\end{cases},\label{Eq6}
\end{align}
and
\begin{align}\sum_{k=0}^{N-1} \omega^{2k(\ell-i)}= \begin{cases} N ~~~ ;  ~~~~N~|\ell-i\\0 ~~~~ ;  ~~~~N\nmid \ell-i\end{cases}.\label{Eq6}
\end{align}
Given any arbitrary $i$, $\ell+i\in [N_1+i:N_2+i]$ are $N$ consecutive numbers, and only one of them is divisible by $N$. Let $\ell_1, \ell_2$ be unique numbers in $[N_1:N_2]$ such that $N|\ell_1+i$ and $N|\ell_2-i$. We then have
\begin{align*} \lambda_i &=\frac N2 \omega ^{-\ell_1} +\frac N2\omega ^{\ell_2} 
=\frac N2 \omega^{(\ell_2-\ell_1)/2}(\omega ^{-(\ell_2+\ell_1)/2} +\omega ^{(\ell_2+\ell_1)/2})
\\&=N \omega^{(\ell_2-\ell_1)/2}\cos(\frac{(\ell_2+\ell_1)\pi}{2N}).\end{align*}
Thus,
$|\lambda_i|=N|\cos(\frac{(\ell_2+\ell_1)\pi}{2N})|$. From $N|\ell_1+i$ and $N|\ell_2-i$, we have that $N|\ell_1+\ell_2$. But $\ell_1+\ell_2\in [2N_1:2N_2]$ consists of $2N_2-2N_1+1=2N-1$ consecutive natural numbers, and there cannot be more than two numbers that are divisible by $N$ in $[2N_1:2N_2]$. If $\ell_1+\ell_2=kN$, we will have
\begin{align*}|\lambda_i|=N|\cos(\frac{k\pi}{2})|= \begin{cases} N ~~~ ;  ~~~~k\text{ is even} ~~~~ \\0 ~~~~;  ~~~~k\text{ is odd}~~~~\end{cases}.\end{align*}
Therefore, we need to find out for the number of values of $i$, $|\lambda_i|$ is zero, and for the number values of $i$ that $|\lambda_i|$ is $N$. 

Assume that $N_1=qN+r$. Then $N_1+i=qN+r+i$ and $N_1-i=qN+r-i$. Therefore \begin{align}\ell_1=\begin{cases}N_1+N-(r+i) &r+i\leq N
\\
N_1+2N-(r+i) & r+i> N
\end{cases},
\end{align}
\begin{align}\ell_2=\begin{cases}N_1 & i-r=N\\
N_1+i-r &N-1\geq i-r\geq 0
\\
N_1+N-(r-i) & i-r<0
\end{cases},
\end{align}
By considering different cases, one gets that the number of values of $i$ where $|\lambda_i|$ is $N$ is equal to $f(N_1, N)$, where
$$f (a,b)= \begin{cases}
   b-1      & \text{if}  ~ r = 0\\
   2b-2r+1       & \text{if } 2r >b \\
   2r-1  & \text{if } 0<2r \leq b 
  \end{cases}, $$
where $q$ and $r$ are the quotient and remainder of dividing $a$ by $b$. 

Now for $N<M\leq 2N,$ we use Theorem \ref{T4}  from Appendix \ref{sec:B81} for  matrix $\Phi$ with singular values $\delta_1 \geq \delta_2, \cdots \geq \delta_N$ and submatrix $\textbf{G}_{N \times (M-N)}$  with singular values $s_1 \geq s_2 , \cdots \geq s_{M-N}$. Therefore, 
$\delta_i \geq s_i   ~~~   \text{for}   ~~~ i= 1, 2, \cdots, M-N,$
thus, $\textbf{G}_{N \times (M-N)}$ has at most $\mathsf{Num}=\min\big(f(N_1,N), M-N\big)$ non-zero singular values. Furthermore the absolute value of non-zero eigenvalues of $\mathbf{G}$ is  less than or equal to $N$. 
Using this upper bound on $s_i$ and substituting it into  \eqref{Eq7}, we get
\begin{align}\tr(\Delta_1^{-1})&=\frac{2N-M}{N\ps +\sigma^2}+
\sum_{i=1}^{M-N}\frac{2(N\ps +\sigma^2)}{(N\ps +\sigma^2)^2-p^2 s_i^2}\\&\leq \frac{2N-M}{N\ps +\sigma^2}+ \mathsf{Num} \frac{2(N\ps +\sigma^2)}{(N\ps +\sigma^2)^2-p^2 N^2}+(M-N-\mathsf{Num}) \frac{2}{(N\ps +\sigma^2)},\end{align}
Thus,  \eqref{eqn:Distortion2mmmm1} would be 
\begin{align}
2D&=(2N-M)\ps +\ps\sigma^2\cdot\tr(\Delta_1^{-1})
\\&\leq(2N-M)\ps +\ps\sigma^2\bigg\{\frac{2N-M}{N\ps +\sigma^2}+\mathsf{Num} \cdot\frac{2(N\ps +\sigma^2)}{(N\ps +\sigma^2)^2-p^2 N^2}+(M-N-\mathsf{Num}) \cdot\frac{2}{(N\ps +\sigma^2)}\bigg\},
\end{align}
which results in the following upper bound
$$\frac{\mathsf{D}_{\min}}{\ps}\leq \frac12(2N-M)+ \frac{2N-M}{2(1+SNR)}+\mathsf{Num}\cdot \frac{ 1+SNR}{1+2~ SNR }+(M-N-\mathsf{Num})\cdot \frac{1}{ 1+SNR}. $$
This concludes the proof of the upper bound.

\subsubsection{Results Needed for the Proof of Theorem \ref{T:NM2N}}\label{sec:B81}
\begin{theorem}\label{L3}\cite[Ex. 6, 17-2]{Hogben}
Consider the Jordan--Wielandt matrix of the block form
$P=\left[\begin{matrix}
\mathbf{0}&G_{m\times n}\\
{G}^{\dagger}&\mathbf{0}
\end{matrix}\right].$ Assume that $\{s_i(G)\}$ are the  singular values of  $G$. Then, the eigenvalues of $P$ are $\{\pm s_i(G)\}$ together with $|m-n|$ zeros.
\end{theorem}  
\begin{theorem}\label{T2}\cite[p. 34]{Gray}
Every  $n \times n$ circulant matrix~ $C= [c_{k-j} ] $ has eigenvectors $y^{(m)}=\frac1{\sqrt{n}}[1 , e^{ -j2\pi m/n} , ..., e^{-  j2\pi(n-1)m/n}],$ for $m = 0, 1, \cdots, n-1$ and corresponding eigenvalues $\psi_m=\sum_{k=0}^{n-1}c_k e^{- j2\pi mk/n},$ and can be expressed in the form $C=U\Psi U^*$, where $U$ has the eigenvectors as columns in order and $\Psi =diag(\psi_m)$ is a diagonal matrix with diagonal elements  $\psi_0, \psi_1, \cdots, \psi_n-1$.
\end{theorem}

\begin{theorem}\label{T3}\cite[p. 35]{Gray}
Let $B=[b_{k-j}]$ and $C=[c_{k-j}]$ be two $n \times n$ circulant matrices with eigenvalues 
$\beta_m=\sum_{k=0}^{n-1}b_k e^{ -j2\pi mk/n},~~ \psi_m=\sum_{k=0}^{n-1}c_k e^{ -j2\pi mk/n},$
 respectively. Then matrices $B$ and $C$ commute and $BC=CB=U\Psi U^*,$ where $\Psi= diag(\beta_m\psi_m ),$ and $BC$ is also a circulant matrix.
\end{theorem}

\begin{theorem}\label{T4}\cite{Thopmson}
Let $A$ be an $m \times n$ matrix with singular values 
$\alpha_1\geq \alpha_2 \geq \cdots \geq \alpha_{\min(m,n)}.$
Let $B$ be a $ p\times q$ submatrix of $A$ (intersection of any $p$ rows and any $q$ columns of $A$) with singular values 
$\beta_1\geq \beta_2 \geq \cdots \geq \beta_{\min(p,q)}.$
Then,\\
$$\begin{cases}
\alpha_i \geq \beta_i  & \emph{for} ~~~ i = 1, 2, \cdots, \min(p,q)\\
 \beta_i \geq \alpha_{i+(m-p)+(n-q)} & \emph{for}~~~  i \leq \min(p+q-m, p+q-n)
\end{cases}.$$
\end{theorem}

\subsection{Proof of Theorem \ref{thm:unifMlessN} (Section \ref{sec:filter-Half})}\label{proof:4:2:2}
Since we have a pre-sampling filter  $H(\omega)$, we use the average distortion formula given in \eqref{eqn:Distortion1}:
\begin{align}2\mathsf{D}&=(2N-M)\ps +\ps\sigma^2\tr\left((\ps QLL^\dagger Q^\dagger + \sigma^2 I)^{-1}\right)
\\&=M\ps +\ps\sigma^2\tr\left((\ps QLL^\dagger Q^\dagger + \sigma^2 I)^{-1}\right)
\end{align}
One can verify that the diagonal entries of the matrix $\ps QLL^\dagger Q^\dagger + \sigma^2 I$  are $M\ps/2+\sigma^2$, but this matrix is not diagonal as
\begin{align*}(\ps QLL^\dagger Q^\dagger + \sigma^2 I)_{(2,1)}&=\ps \sum_{\ell=N_1}^{N_2}a_\ell \cos(\ell\omega_0(t_2-t_1))
=\ps \sum_{\ell=N_1}^{N_1+M/2-1}\cos(\frac{\ell\omega_0T}{M})
\\&=\ps \sum_{\ell=N_1}^{N_1+M/2-1}\cos(\frac{2\pi\ell}{M})
=-\ps \frac{\sin\left(\frac{(2N_1-1)\pi}{2M}\right)}{\sin(\frac{\pi}{M})}\neq 0,
\end{align*}
 since $M$ is even. 
Hence the equality condition in the Theorem \ref{thm:lemma:diag} is not  satisfied and
\begin{align}\mathsf{D}_{H} &= \frac12 \ps \left(M+\sigma^2\tr\left((\ps QLL^\dagger Q^\dagger + \sigma^2 I)^{-1}\right)\right)
\nonumber\\& >
\frac12 \ps \left(M+\sigma^2 \frac{M}{\frac{M\ps}{2}+\sigma^2}\right)
=\frac12 \ps \left(M+\frac{M}{1+\frac{1}{2} SNR}\right).\nonumber
\end{align}

\subsection{Proof of Proposition \ref{Proposition:Sparse} (Section \ref{sparsesection})}\label{proof:4:5:2}
To study non-uniform power constraints, we can use the same mathematical framework given in Section \ref{sec:matrix-form}. The only change is that $C_\X$ is no longer equal to $\ps I$. Here 
$C_\X $ is a diagonal matrix with the following diagonal entries
\begin{align}C_{\X(k,k)} = \begin{cases} \ps_{k+N_1-1}~~~~~~~;  ~~~~k = 1, \cdots, N
\\\ps_{k+N_1-1-N}~~~ ;  ~~~~k = N+1, \cdots, 2N\end{cases}.\label{Cxsparse}
\end{align}
Therefore, \eqref{Ce2half} results in
\begin{align}
C_e & = (Q^\dagger C_{\Z}^{-1}  Q+C_{\X}^{-1}  )^{-1}= (\frac{1}{\sigma^2 } Q^\dagger Q+C_{\X} ^{-1} )^{-1}=\ps\sigma^2 (\ps Q^\dagger Q+\ps\sigma^2 C_{\X} ^{-1} )^{-1}. \end{align}
Let $\Delta'_2=Q^\dagger Q+\sigma^2 C_{\X} ^{-1}$. Observe that when $\ps_\ell=\ps$ for all $\ell$, $A'_2$ reduces to $\Delta_2$. In fact, since $C_{\X} ^{-1}$ is a diagonal matrix, the proof of Lemma \ref{lemma:lower:2} can be directly mimicked here, with no essential changes. Similarly,  the proof of Theorem~\ref{T1c} can be adapted in a straightforward manner to show that uniform sampling, \emph{i.e.,}  $t_i=iT/{M}, i=1, 2, \cdots, M$ is a solution to the above equation if for each $k$ in the interval $ 2N_1 \leq k \leq 2N_2$, $M$ does not divide $k$. Particularly, for the rates above the Nyquist rate, \emph{i.e.,}  $M>2N_2$, uniform sampling is optimal.


\begin{thebibliography}{10}

\bibitem{timestampless}
S. Feizi, V. K. Goyal and M. Medard, ``Time-stampless adaptive nonuniform sampling for stochastic signals",  IEEE Trans. Signal Process., vol. 60, no. 10, pp. 5440-5450, Oct. 2012.

\bibitem{Walden}
 R.H. Walden, ``Analog-to-digital converter survey and analysis," IEEE  J.  Sel. Areas Commun., vol. 17, no. 4, pp. 539-550, Apr. 1999.

\bibitem{Bretthorst}
G. L. Bretthorst. (2008, Nov) ``Nonuniform sampling: Bandwidth and aliasing."  Concepts in Magnetic Resonance (Part A). [Online]. 32A (6), pp. 417-435. Available: www.onlinelibrary.wiley.com/doi/10.1002/cmr.a.v32a:6/issuetoc.

\bibitem{Landau}
H. J. Landau, ``Necessary density conditions for sampling and interpolation of certain entire functions," Acta Math., vol. 117, pp. 37-52, 1967.

%\bibitem{F.Marvasti}
%F. Marvasti, \emph{A unified approach to zero- crossings and nonuniform sampling}, Nonuniform, (1987).

\bibitem{Marvasti}
F. Marvasti, \emph{Nonuniform Sampling, Theory and Practice}, Kluwer Academic/Plenum Publishers, New York, 2000.

\bibitem{Lin} Y.-P. Lin and P. Vaidyanathan, ``Periodically nonuniform sampling of bandpass signals," IEEE Trans. Circuits Syst. II, vol. 45, no. 3, pp. 340-351, Mar 1998.

%\bibitem{EldarOp}
%Y. C. Eldar and A. V. Oppenheim, ``Filterbank reconstruction of bandlimited signals from nonuniform and generalized samples," IEEE Trans. Signal Process., vol. 48, no. 10, pp. 2864-2875, 2000.


\bibitem{Bhatia} R. Bhatia, \emph{Matrix analysis,}  Springer Science \& Business Media, 1997.

\bibitem{eldarq}
A. Kipnis, A. Goldsmith, Y. Eldar and T. Weissman, ``Distortion Rate Function of Sub-Nyquist Sampled Gaussian Sources," 
 IEEE Trans. Inf. Theory , vol. 62, no. 1, pp. 401-429,  Jan. 2016.



\bibitem{Balakrishnan}
A. Balakrishnan, ``A note on the sampling principle for continuous signals,"  IRE Trans. Inf. Theory, vol. 3, no. 2, pp. 143-146, Jun. 1957.


\bibitem{Matthews}
M. Matthews, ``On the linear minimum-mean-squared-error estimation of an undersampled wide-sense stationary random process,"  IEEE Trans.  Signal
Process., vol. 48, no. 1, pp. 272-275, Jan. 2000.

\bibitem{Chan}
D. Chan and R. Donaldson, ``Optimum pre-and postfiltering of sampled signals with application to pulse modulation and data compression systems," IEEE Trans. Commun. Techn., vol. 19, no. 2, pp. 141–157, Apr. 1971.



\bibitem{WuSun}
J. Wu and N. Sun. ``Optimum sensor density in distortion-tolerant wireless sensor networks,"  IEEE Trans. Wireless Commun., vol. 11, no. 6, pp. 2056-2064, Jun. 2012.

%\bibitem{KipnisNew}
% A. Kipnis, A. Goldsmith, and Y. C. Eldar, ``Gaussian distortion-rate function under sub-Nyquist nonuniform sampling," in 52st Annual Allerton Conference on Communication, Control, and Computing (Allerton), 2014. 



\bibitem{app1}
L. J. Alvarez-Vazquez, A. Martinez, M. E. Vazquez-Mendez and M. A. Vilar, ``Optimal location of sampling points for river pollution control," Mathematics and Computers in Simulation, vol. 71, no. 2, pp. 149-160, Feb. 2006.

\bibitem{app3}
R. Wang, and H. Zhang, ``Optimal Sampling Points in Reproducing Kernel Hilbert Spaces," 2012.  [Online]. Available: http://arxiv.org/abs/ 1207.5871.

\bibitem{app4}
S. Rajendran and K. M. Liew, ``Optimal stress sampling points of plane triangular elements for patch recovery of nodal stresses," International J. for numerical methods in engineering, vol. 58, no. 4, pp. 579-607, Apr. 2003.


\bibitem{Prakash}
V. P. Boda and P. Narayan, ``Sampling rate distortion," in Proc. IEEE Int. Symp. on Inf. Theory (ISIT), 2014, pp. 3057-3061.

\bibitem{Davis}
C. Guo, and M. E. Davies, ``Sample Distortion for Compressed Imaging," IEEE Trans. Signal Process., vol. 61, pp. 6431--6442, Dec. 2013.

\bibitem{Gastpar}
G. Reeves and M. Gastpar, ``The Sampling Rate-Distortion Tradeoff for Sparsity Pattern Recovery in Compressed Sensing", IEEE Trans. Inf. Theory, vol. 58, no. 5, pp. 3065-3092, May 2012.


%\bibitem{Oppenheim}
%A. V. Oppenheim, G. C. Verghese, \emph{Signals, Systems and Inference}, Pearson, 2015.

\bibitem{Luenberger}
D. G. Luenberger, ``Optimization by vector space methods," John Wiley \& Sons, 1990.


\bibitem{Eric}
E. Carlen, ``Trace inequalities and quantum entropy: an introductory course," Entropy and the Quantum, vol. 529, pp. 73-140, 2010.

\bibitem{Hogben}
L. Hogben, \emph{Handbook of linear algebra}, CRC Press, 2006.

\bibitem{Gray}
R. M. Gray, \emph{Toeplitz and circulant matrices: A review}, Now publishers Inc., 2006.

\bibitem{Thopmson}
R. C. Thompson, ``Principal submatrices IX: Interlacing inequalities for singular values of submatrices," Linear Algebra and its Applications, vol. 5, no. 1, pp. 1-12, Jan. 1972.


\end{thebibliography}
\end{document}